\pdfoutput=1
\documentclass[conference]{IEEEtran}

\usepackage{versions}

\includeversion{full}
\excludeversion{conf}
\includeversion{maybe}
\includeversion{nonanonymous}

\InputIfFileExists{localflags}{} 

\newcommand{\fullversionref}{accunbounded-full}
\usepackage{xr} 
\newcommand{\appendixorfull}[1]{%
\processifversion{conf}{the full version~\cite[\Cref{F-#1}]{\fullversionref}}%
\processifversion{full}{\Cref{#1}}%
}

\usepackage{pifont} 
\newcommand{\xmark}{\ding{55}}%

\usepackage{custom}
\usepackage{cryptostyle}

\hypersetup{
  pdfauthor   = {Kevin Morio,Robert Künnemann},%
  pdftitle    = {Verifying Accountability for Unbounded Sets of Participants},%
  pdfsubject  = {Computer Science},%
  pdfkeywords = {accountability, formal methods, verification, security protocols}%
}

\author{}

\begin{nonanonymous}
  \author{%
  \IEEEauthorblockN{Kevin Morio and Robert Künnemann}
  \IEEEauthorblockA{CISPA Helmholtz Center for Information Security \\
                    Saarland Informatics Campus, Germany}}
\end{nonanonymous}

\IEEEoverridecommandlockouts

\title{Verifying Accountability for \\ Unbounded Sets of Participants
    \processifversion{full}{\\(Full Version)\thanks{This is an extended version of \cite{morio2021verifying}.}}
}
\date{}

\begin{document}

\maketitle

\begin{abstract}
Little can be achieved in the design of security protocols without trusting at least some participants.
This trust should be justified or, at the very least, subject to examination.
One way to strengthen trustworthiness is to hold parties accountable for their actions, as this provides a strong incentive to refrain from malicious behavior.
This has led to an increased interest in accountability in the design of security protocols.

In this work, we combine the accountability definition of \textcite{kuennemann2019automated} with the notion of case tests to extend its applicability to protocols with unbounded sets of participants.
We propose a general construction of verdict functions and a set of verification conditions that achieve soundness and completeness.

Expressing the verification conditions in terms of trace properties allows us to extend \Tamarin{}---a protocol verification tool---with the ability to analyze and verify accountability properties in a highly automated way.
In contrast to prior work, our approach is significantly more flexible and applicable to a wider range of protocols.
\end{abstract}

\section{Introduction}
\label{sec:introduction}
\label{sec:overview}

Holding parties accountable for their misconduct---most often detection is deterrent enough---is an incentive to avoid malicious behavior from the outset.
Participants have to weigh up whether an action is worth the consequences.
Accountability is applicable to a wide range of protocols, such as e-voting, electronic payment processing, and electronic health care transactions.

\textcite{kuennemann2019automated} proposed a protocol-agnostic definition of accountability
and an automated verification technique.
They consider accountability a meta-property defined with respect to a security property $\secprop$.
A protocol that provides accountability for $\secprop$ provides the information necessary to determine whether $\secprop$ has been violated and, if so, which parties should be held accountable.
The \emph{verdict} contains
all groups of parties that are (jointly) accountable for a violation.
Verdicts are returned by a total function---the \emph{verdict function}---given the trace of a protocol execution.

To prove that a specified verdict function provides the protocol with accountability for a security property, a set of verification conditions must be verified.
\textcite{kuennemann2019automated} show that their proposed verification conditions are sound and complete:
If and only if all conditions hold, accountability is provided.
Expressing the verification conditions as trace properties allows them to exploit existing protocol verification tools and achieve a high degree of automation.

However, the verdict function and verification conditions they propose
require the parties to be explicitly stated in each verdict,
thus inherently limiting the set of parties that can be blamed.
This restricts the expressiveness of the approach as in almost all
real-world protocols the same party can be involved in multiple parallel sessions (e.g., TLS) or the number of participants is not known a priori and can change dynamically during the protocol execution (e.g., the Signal chat protocol).%

In this work, we address this shortcoming by combining their approach with the notion of case tests---an idea inspired by the accountability tests of \mbox{\textcite{bruni2017automated}}.
Case tests are trace properties with free variables,
where each free variable stands for a party that should be blamed for a violation.
This is in contrast to accountability tests, where one test applies to
one party and joint accountability is not expressible.
In contrast to an explicitly stated verdict function, case tests can
match multiple parts of a trace.
The verdict function is thus implicitly defined as the union of
groups of parties blamed by instances of the case tests. 
This approach provides flexibility and allows---for the first time---
the analysis of protocols with unbounded sets of participants.
It also improves readability, as intuitively, each case test stands
for a specific way a violation can be triggered, in contrast to the
previous, explicit formulation of the verdict function, in which all
combinations of actions that constitute a violation had to be captured.

We list our contributions as follows.
\label{sec:contributions}

\begin{enumerate}
  \item We derive a set of verdict-based verification conditions based
      on the accountability definition of
      \textcite{kuennemann2019automated} and show that they provide
      soundness and completeness.
  \item We introduce the notion of case tests and use them to define
      verdict functions that are significantly more flexible than
      previous ones.
  \item We show how the verification conditions can be rewritten using
      case tests and formalize requirements to encode them in terms of
      trace properties.
      By proving their relation to the verdict-based variant, we can
      transfer their
      soundness and completeness to the
      encoded verification conditions.
  \item We implemented our approach in \Tamarin{} by adding the ability
      to define accountability lemmas and case tests.
  \item We showcase our methodology by extending previous models of
      OCSP Stapling~\cite{rfc6960,rfc6066}, and
      Certificate Transparency~\cite{rfc6962}
      to an unbounded
      number of parties,
      and applying it to new models of mixnets
      and the e-voting protocol MixVote/Alethea~\cite{basin2018alethea,DBLP:conf/csfw/BasinRS20}.
\end{enumerate}

The paper is structured as follows.
Related work is discussed in \cref{sec:related-work} and the accountability definition of \textcite{kuennemann2019automated} is elucidated in \cref{sec:background}.
We show a sound and complete decomposition of accountability into verification conditions in \cref{sec:verif-acc} and
introduce case tests to define a general verdict function in \cref{sec:verdict-functions}.
We elaborate on the counterfactual relation in
\cref{sec:count-rel} and present the verification conditions expressed as trace properties
in \cref{sec:verification-conditions-trace-prop}.
In \cref{sec:automated-verif}, we explain the implementation in
\Tamarin{}.
In \cref{sec:case-studies}, we describe the case studies, evaluate our verification results, and compare them with the results in the framework of \cite{kuennemann2019automated}.
We conclude in \cref{sec:conclusion}.

\section{Related Work}
\label{sec:related-work}

In the security setting considered in this work, we regard accountability as the ability to identify malicious parties.
Different approaches and notions of accountability have been proposed.
However, in previous works, these are only described informally or tailored to specific protocols and security properties \parencites{asokan1998asynchronous,backes2005anonymous,backes2013privacy,kroll2015phd}.
An emerging problem is the difficulty of defining when a party's behavior should be considered \emph{malicious} and the implications this has on completeness---holding \emph{all} malicious parties accountable.%
\footnote{
  Note that this notion of completeness concerns what we would consider the verdict.
  This is opposed to the completeness of the verification conditions with respect to the accountability of a protocol, saying that all protocols that provide accountability are recognized as correct by the verification conditions.
}

In the past, misbehaving and dishonest parties were treated as equivalent.
While this is a reasonable approximation for some cryptographic tasks---for example, secure multi-party computation---it is not suitable in the context of accountability.
Completeness would require identifying all dishonest parties, but a dishonest party does not have to deviate or may behave in a way that is indistinguishable from the protocol.
Some approaches \parencites{haeberlen2007peerreview,jagadeesan2009towards} in the distributed setting assume that all communication is observable and classify any trace not producible by honest parties as malicious behavior.
In the security setting, this assumption is impossible to satisfy and the definition of malicious behavior unreasonable, as parties may communicate through hidden channels and deviate in harmless ways.

\textcite{haeberlen2007peerreview} propose PeerReview, a system which can detect Byzantine faults in the distributed setting.
The system requires that all communication is observable which is a suitable assumption in a distributed environment but unrealistic in the security setting.

\textcite{jagadeesan2009towards} provide multiple general notions of accountability based on an abstract labeled transition system in the distributed setting.
However, the authors admit that \enquote{the only auditor capable of providing [completeness] is one which blames all principals who are capable of dishonesty, regardless of whether they acted dishonestly or not.}

\textcite{kusters2010accountability} define accountability in the symbolic and computational model using accountability properties.
These are specified in a formal language.
\textcite{kuennemann2018accountability} point out that these properties are not expressive enough.
Furthermore, they identify significant weaknesses in the case of joint misbehavior.

Another approach is to consider protocol actions as the actual causes for security violations \parencites{feigenbaum2011towards,datta2015program,gossler2015general}.
However, protocol actions may be causally related to a security violation but still be harmless and without any malicious intent.

In recent work, \textcite{kuennemann2018accountability} propose a general protocol agnostic definition of accountability in which the fact that a party deviated is considered a potential cause for a security violation.
Based on this approach, \textcite{kuennemann2019automated} provide an automated verification technique in the single-adversary setting.
They define the \emph{a posteriori verdict} (apv), which given a trace of a protocol execution, returns all groups of parties that are jointly accountable for the security violation.
If there exists a function---called the verdict function---which coincides with the apv for all traces of the protocol, the verdict function is said to provide the protocol with accountability for a specified security property.
In their work, the verdict function uses a case distinction over traces and specifies a verdict per case.
This form requires that the parties be explicitly stated in a verdict and thus fixes the number of parties.
Most protocols have a fixed number of \emph{roles}, but the same party
can run many sessions with different communication partners (e.g.,
several servers in TLS, or partners in chat protocols).

\textcite{bruni2017automated} give a definition of accountability based on the existence of per-party accountability tests which decide whether the party should be held accountable for a violation.
A verdict is obtained by considering all parties for which their test is positive as singleton sets.
Since each singleton verdict contains exactly one party, joint accountability is not expressible. 
Moreover, as noted by \textcite{kuennemann2019automated}, there are some flaws in the criteria of their definition allowing parties to be blamed even if the security of a protocol cannot be violated or violations remain undetected under certain circumstances.

\section{Background}
\label{sec:background}

We provide an overview of the notation and concepts we use throughout
this work and recall the accountability definition of \textcite{kuennemann2019automated}.
\subsection{Preliminaries}
\label{sec:notation}

\paragraph{Sets, sequences, and multisets}
We denote the set of integers $\Set{ 1, \dots, n }$ by $[n]$,
the power set of $S$ by $2^S$ and the set of finite sequences of elements from $S$ by $S^*$.
For a sequence $s$, we write $s_i$ for the $i$-th element, $\abs{ s }$ for the length of $s$, and $\idx(x) \defas \Set[\big]{1, \dots, \abs{ s } }$ for the set of indices of $s$.
We write $\vec{s}$ to emphasize that $s$ is a sequence.
\begin{full}
For a set $A$, we write $A^{\#}$ for the set of finite multisets of elements from $A$.
We use the superscript $^{\#}$ to denote the usual operations on multisets.
For example, we write $\emptyset^{\#}$ for the empty multiset and $m_1 \union^{\#} m_2$ for the union of two multisets $m_1$ and $m_2$.
\end{full}

\paragraph{Terms}
Cryptographic messages are modeled as abstract terms.
We specify an order-sorted term algebra with the sort $\sort{msg}$ and two incomparable subsorts $\sort{pub}$ and $\sort{fresh}$ for two countably infinite sets of public names ($\PN$) and fresh names ($\FN$).
We assume pairwise disjoint, countably infinite sets of variables $\Vars_s$ for each sort $s$.
The set of all variables $\Vars$ is the union of the set of variables for all sorts $\Vars_s$.
We write $u \colon s$ when the name or variable $u$ is of sort $s$.
A signature $\Sig$ is a set of function symbols, each with an arity.
We write $f / n$ for a function symbol $f$ with arity $n$.
\begin{full}
A subset $\Sig_{\text{priv}} \subseteq \Sig$ consists of \emph{private} function symbols which cannot be applied by the adversary.
\end{full}
The set of well-sorted terms constructed over $\Sig$, $\PN$, $\FN$, and $\Vars$ is denoted by $\Terms_{\Sig}$.
The subset of ground terms---terms without variables---is denoted by $\GTerms_{\Sig}$.
If $\Sig$ can be inferred from context, we write $\Terms$ and $\GTerms$ respectively.

\paragraph{Equational theories}
An \emph{equation} over the signature $\Sig$ is an unordered pair $\Set{ s, t }$ of terms $t$, $s \in \Terms_{\Sig}$, written $s \backsimeq t$ or $s = t$ when the meaning can be inferred from context.
Equality is defined with respect to an equational theory $E$, a binary relation $=_E$ induced by a finite set of equations which is closed under the application of function symbols, bijective renaming of names, and substitution of variables by terms of the same sort.
An equational theory $E$ formalizes the semantics of the function symbols in $\Sig$.
We say that two terms $t$ and $s$ are \emph{equal modulo $E$} iff $t =_E s$.
Set membership modulo $E$ is denoted by $\in_E$ and defined as $e \in_E S$ iff $\Exists{ e' \in S } e' =_E e$.
The usual operations on sets modulo $E$ are defined accordingly.

\begin{full}
\begin{example}[Digital signatures]
  \label{ex:signatures}
  To model cryptographic messages built using digital signatures, we use the signature
  \begin{equation*}
    \Sig_{\DS} = \Set{ \fun{sig}/2, \fun{verify}/3, \fun{pk}/1, \fun{sk}/1, \fun{true}/0 }
  \end{equation*}
  and the equational theory $E_{\DS}$ generated by the equation
  \begin{equation*}
    \fun{verify}\del[\big]( \fun{sig}( \var{m}, \fun{sk}( \var{i} ) ), \var{m}, \fun{pk}( \fun{sk}( \var{i} ) ) ) = \fun{true}\,. \qedhere
  \end{equation*}
\end{example}

We assume that the signature $\Sig$ and the equational theory $E$ contain symbols and equations for pairing and projection.
  $\Set[\big]{ \del< \cdot, \cdot >, \fun{fst}/1, \fun{snd}/1 } \subseteq \Sig $
  with
  \begin{align*}
  \fun{fst}\del[\big]( \del< \var{x}, \var{y} > ) &= x,
                                                  &
    \fun{snd}\del[\big]( \del< \var{x}, \var{y} > ) &= y\,.
  \end{align*}
We write $\langle \var{x}_1, \langle \dots, \langle \var{x}_{n-1}, \var{x}_n \rangle \dots \rangle$ simply as $\del< \var{x}_1, \dots, \var{x}_n >$.
\end{full}

\paragraph{Facts}
We assume an unsorted signature $\Sig_{\FACT}$ which is disjoint from $\Sig$.
The set of \emph{facts} is defined by
\begin{equation*}
  \Facts \defas \DisplaySet[\big]{ \fact{F}(t_1, \dots, t_n) \given t_i \in \Terms, \fact{F} \in \Sig_{\FACT}^k }\,,
\end{equation*}
where $\Sig_{\FACT}^k$ denotes all function symbols of arity $k$ in $\Sig_{\FACT}$.
The set of \emph{ground facts} is denoted by $\GFacts$.

\paragraph{Substitutions}
A \emph{substitution} $\sub$ is a well-sorted function from variables $\Vars$ to terms $\Terms_{\Sig}$ that corresponds to the identity function on all variables except on a finite set of variables.
Overloading notation, we call this finite set of variables the \emph{domain} of $\sub$, which we denote by $\dom(\sub)$.
The image of $\dom(\sub)$ under $\sub$ is denoted by $\rng(\sub)$.
For the homomorphic extension of $\sub$ to a term $t$ or a trace formula $\varphi$, we write $t \sub$ and $\varphi \sub$ respectively.
We write $\sub \subst{ v \mapsto w }$ to denote the update of $\sub$ at $v$ such that $\sub\subst{v \mapsto w}(x) = w$ for $x = v$ and $\sub(x)$ otherwise.
If $\sub$ is injective, we denote its inverse by $\sub^{-1}$.
We say that two substitutions $\sub$, $\sub'$ are equal modulo $E$ if $\dom(\sub) = \dom(\sub')$ and $\sub(x) = \sub(x')$ for all $x$ in $\dom(\sub)$.

\paragraph{Valuation}
Each sort $s$ is associated with a domain $\Dom_s$.
The domain for temporal variables is the rational numbers $\Dom_{\sort{temp}} \defas \QQ$ and the domains for messages are $\Dom_{\sort{msg}} \defas \GTerms$, $\Dom_{\sort{fresh}} \defas \FN$, and $\Dom_{\sort{pub}} \defas \PN$.
A function $\valu$ from $\Vars$ to $\QQ \union \GTerms$ is a \emph{valuation} if it respects sorts, that is, $\valu(\Vars_s) \subseteq \Dom_s$ for all sorts $s$.
We write $t \valu$ for the homomorphic extension of $\valu$ to a term $t$.

\paragraph{Trace properties}
Trace properties are sets of traces which are specified by trace formulas in a two-sorted first-order logic which supports quantification over messages and timepoints.
\begin{definition}[Trace formula]\label{def:trace-formula}
  A trace atom is either false $\bot$, a term equality $t_1 \approx t_2$, a timepoint ordering $i \lessdot j$, a timepoint equality $i \doteq j$, or an action $\fact{F}@i$ for a fact $\fact{F} \in \Facts$ and a timepoint $i$.
  A trace formula is a first-order formula over trace atoms.\looseness=-1
\end{definition}

\begin{definition}[Satisfaction relation]
  \label{def:satis-rel}
  The satisfaction relation $(\tr, \valu) \holds \varphi$ between a trace $\tr$, a valuation $\valu$, and a trace formula $\varphi$ is defined as follows.
  \begingroup
  \allowdisplaybreaks
  \begin{align*}
    &(\tr, \valu) \holds \bot                          &&\phantom{{}\iff{}} \text{never} \\
    &(\tr, \valu) \holds \fact{F}@i                    &&\iff \valu(i) \in \idx(\tr) \land \fact{F} \valu \in_E \tr_{\valu(i)} \\
    &(\tr, \valu) \holds i \lessdot j                  &&\iff \valu(i) < \valu(j) \\
    &(\tr, \valu) \holds i \doteq j                    &&\iff \valu(i) = \valu(j) \\
    &(\tr, \valu) \holds t_1 \approx t_2               &&\iff t_1 \valu =_E t_2 \valu \\
    &(\tr, \valu) \holds \neg \varphi                  &&\iff \text{not } (\tr, \valu) \holds \varphi \\
    &(\tr, \valu) \holds \varphi_1 \land \varphi_2     &&\iff (\tr, \valu) \holds \varphi_1 \text{ and } (\tr, \valu) \holds \varphi_2 \\
    &(\tr, \valu) \holds \Exists{ x \colon s } \varphi &&\iff \begin{aligned}[t]
                                                               &\text{there is } u \in \dom(s) \\
                                                               &\text{such that } (\tr, \valu \subst{ x \mapsto u }) \holds \varphi\,.
                                                              \end{aligned}
  \end{align*}
  \endgroup
\end{definition}
For completeness, we define
\begin{align*}
  &\varphi_1 \lor \varphi_2 &&\equiv \neg ( \neg \varphi_1 \land \neg \varphi_2 ) \\
  &\Forall{ x \colon s } \varphi &&\equiv  \neg ( \Exists{ x \colon s } \neg \varphi ) \\
  &\varphi_1 \impliestp \varphi_2  &&\equiv  \neg \varphi_1 \lor \varphi_2 \\
  &\varphi_1 \ifftp \varphi_2      &&\equiv \varphi_1 \impliestp \varphi_2 \land \varphi_2 \impliestp \varphi_1
\end{align*}

We write $t_1 = t_2$, $i < j$, and $i = j$ when the meaning is clear from context.
Timepoints are used to indicate the position of facts in a trace.
The free variables of $\varphi$ are denoted by $\fv(\varphi)$ which may be used as a set or sequence depending on the context.
We say $\varphi$ is a \emph{ground formula} if it does not contain free variables, that is, $\fv(\varphi) = \emptyset$.
When $\varphi$ is a ground formula, we may write $\tr \holds \varphi$ since the satisfaction of $\varphi$ is independent of the valuation.
The renaming of the free variables of $\varphi$ by a sequence $\vec{v}$ of equal length is denoted by $\varphi\subst{ \fv(\varphi) \mapsto \vec{v} }$ or simply $\varphi\subst{ \vec{v} }$.
We write $\varphi\del[\big](\vec{x})$ to denote that the variables $\vec{x}$ are bound in $\varphi$, that is, $\fv\del[\big](\varphi(\vec{x})) = \fv(\varphi) \setminus \vec{x}$.

\begin{full}
\end{full}
\begin{definition}[Validity, satisfiability]
  Let $\Tr$ be a set of traces.
  A trace formula $\varphi$ is \emph{valid} for $\Tr$, written $\Tr \holds^{\forall} \varphi$, iff $(\tr, \valu) \holds \varphi$ for every trace $\tr \in \Tr$ and every valuation $\valu$.
  A trace formula $\varphi$ is \emph{satisfiable} for $\Tr$, written $\Tr \holds^{\exists} \varphi$, iff there exists a trace $\tr \in \Tr$ and a valuation $\valu$ such that $(\tr, \valu) \holds \varphi$.
\end{definition}
Note that $\Tr \holds^{\forall} \varphi$ iff $\Tr \not\holds^{\exists} \neg \varphi$.

\paragraph{Instantiations}
An \emph{instantiation} $\inst$ is a substitution from variables $\Vars$ to ground terms $\GTerms$.
We say that $\inst$ is \emph{grounding} with respect to a trace formula $\varphi$ if $\varphi \inst$ is a ground formula.
For two instantiations $\inst$, $\inst'$ we only consider equality modulo $E$  and simply write $\inst = \inst'$.
In particular, all operations involving instantiations are considered modulo $E$.

\paragraph{Accountability protocol}

The conditions we derive are independent of the formalism of
choice. For now, we assume a function $\traces$ from protocols to sets
of ground traces, i.e., sequences of ground facts.
Given a protocol $\ps{P}$ (e.g., a ground process or a set of multiset-rewrite rules),
$\varphi$ is valid for $\ps{P}$, written $\ps{P} \holds^{\forall}
\varphi$, if $\traces(\ps{P}) \holds^{\forall} \varphi$; $\varphi$ is
satisfiable for $\ps{P}$, written $\ps{P} \holds^{\exists} \varphi$,
if $\traces(\ps{P}) \holds^{\exists} \varphi$.

An accountability protocol is a protocol $\ps{P}$ with a countably infinite set of participants $\Parties$.
We assume that $\Parties \subseteq \GTerms$, that is, a party can be any ground term.
Due to the huge variety in the design of protocols, we leave the concrete structure of this process open.
However, we require that each party which is not trusted specifies a corruption procedure that emits a $\fact{Corrupted}$ fact and reveals its secrets.
The set of corrupted parties of a trace $t$ is defined by
\begin{equation*}
  \cor(t) \defas \DisplaySet[\big]{ \Party{A} \in \Parties \given \fact{Corrupted}(\Party{A}) \in t }\,.
\end{equation*}

In this work, we implicitly assume an accountability protocol $\ps{P}$.
If not stated otherwise, quantification over traces is always with respect to $\traces(\ps{P})$.

\subsection{A Definition of Accountability}
\label{sec:def-acc}

We review the accountability definition of
\textcite{kuennemann2019automated}.
This definition holds parties accountable for violations of a
\emph{security property} which is expressed as a trace property $\secprop$.
To allow any meaningful analysis, there have to be at least two traces, one satisfying and one violating the security property.
Following intuition, if all parties adhere to the protocol, the security property $\secprop$ must hold.
Otherwise, either the protocol or $\secprop$ is ill-defined.

If a violation occurred, i.e., $\neg \secprop$, at least one party must have deviated from the protocol.
Each party is either \emph{honest} and follows the protocol or \emph{dishonest} and may deviate from its specified behavior.
The definition assumes a single adversary controlling all dishonest parties 
(see~\cite{accountability-decentralised} for a discussion of this
topic).
An honest party becomes dishonest when it receives a \emph{corruption} message from the adversary
and remains dishonest for the rest of the protocol execution.
We may refer to parties as dishonest or corrupted interchangeably throughout this work.

A dishonest party does not have to deviate and may behave in a way that is indistinguishable from the protocol.
It is thus impossible to detect all dishonest parties.
Furthermore, parties may deviate by communicating through hidden channels and thus it is also impossible to detect all deviating parties.
Instead, \textcite{kuennemann2019automated} build on sufficient causation \parencites{datta2015program, kuennemann2017sufficient}, and focus on parties that are the actual cause of a violation.
This requires protocols to be defined in such a way that deviating parties leave publicly observable evidence for security violations.
In this sense, a protocol provides accountability with respect to $\secprop$ if we can determine all parties \emph{for which the fact that they are deviating at all} is a cause for the violation of $\secprop$.

Assume a countably infinite set of parties $\Parties$.%
\footnote{In contrast to \cite{kuennemann2019automated} where a finite set of parties is assumed.}
Deviations of a set of parties $S \subseteq \Parties$ are a cause for a violation iff
\begin{description}[labelwidth=3em,leftmargin=3em,labelsep=0pt]
  \item[\itemlabel{SC1}{it:sc-1}:] A violation occurred and the parties in $S$ deviated.
  \item[\itemlabel{SC2}{it:sc-2}:] If all deviating parties, except those in $S$, behaved honestly, the same violation would still occur.
  \item[\itemlabel{SC3}{it:sc-3}:] $S$ is minimal; SC1 and SC2 hold for no strict subset of $S$.
\end{description}

\hyperref[it:sc-1]{SC1} ensures that a violation has occurred and the parties in $S$ deviated.
\hyperref[it:sc-2]{SC2} ensures that the parties in $S$ are sufficient to cause a violation.
There may be other deviating parties not in $S$, but their deviation has no influence on the violation.
In this vein, \hyperref[it:sc-2]{SC2} describes a situation which differs from the actual observed events---called a \emph{counterfactual}.
\hyperref[it:sc-3]{SC3} ensures that only minimal sets $S$ are considered, that is, we always hold the least number of parties accountable.

\begin{example}
  \label{ex:db}
  Consider a protocol in which access to a central user database is logged and each request must be signed.
  A violation occurs whenever user data is leaked.
  The parties involved are a manager $\Party{M}$ and two employees $\Party{E_1}$ and $\Party{E_2}$.
  The manager can directly sign a request to get access to the database and can thus cause a violation on its own.
  For the employees to gain access, both need to sign a request.
  Assume this is the case and a leak occurs.
  Then $\Party{E_1}$ and $\Party{E_2}$ are \emph{jointly} accountable.
  In the counterfactual scenarios in which only $\Party{E_1}$ or $\Party{E_2}$ deviates, a violation is not possible and thus \hyperref[it:sc-2]{SC2} is not satisfied.
  
  If parties $\Party{M}$ and $\Party{E_1}$ cause a violation, then \hyperref[it:sc-1]{SC1} and \hyperref[it:sc-2]{SC2} hold.
  However, as $\Party{M}$ can cause a violation on its own, which also satisfies \hyperref[it:sc-1]{SC1} and \hyperref[it:sc-2]{SC2}, \hyperref[it:sc-3]{SC3} does not hold.
\end{example}

The counterfactual situations considered in \hyperref[it:sc-2]{SC2} cannot be chosen arbitrarily.
They have to be related to the actual situation to obtain meaningful and justifiable results.
This relationship is specified by a \emph{counterfactual relation} $\rel$.
If $(t,t') \in \rel$, also written as $\rel(t,t')$, then the \emph{counterfactual trace} $t'$ is related to the \emph{actual trace} $t$.

We only consider counterfactual traces if they do not
consider additional parties as corrupted, as these are
not being causally relevant for a security violation in the
actual trace.
\begin{definition}[Counterfactual relation]\label{def:counterfactual-relation} 
    A counterfactual relation is a reflexive and transitive relation
    between traces s.t.:
\begin{align}
    \rel(t,t') & \implies \cor(t') \subseteq \cor(t) 
  \label{eq:rel-corrupted}
\end{align}
\end{definition}

The \emph{a posteriori verdict} (apv) specifies for a given trace all minimal subsets of parties that are sufficient to cause a security violation.
\begin{definition}[A posteriori verdict]
  \label{def:apv}
  Let $\ps{P}$ be a protocol, $t$ a trace, $\secprop$ a security property, and $\rel$ a relation on traces.
  The a posteriori verdict is defined by
  $
    \apv_{\ps{P},\secprop,\rel}(t) \defas $
  \begin{align}
    \Big\{\, S\, \SetDelim[\Big] &t \holds \neg \secprop \reqlabel{eq:apv-req-1}\\
                                          {}\land{} &\Exists{ \qs{t'} } \rel(t,t') \land \cor(t') = S \land t' \holds \neg \secprop \reqlabel{eq:apv-req-2}\\
                                          {}\land{} &\NExists{ \qs{t''} } \rel(t,t'') \land \cor(t'') \subsetneq S \land t'' \holds \neg \secprop \,\Big\} \reqlabel{eq:apv-req-3}
  \end{align}
\end{definition}
We may leave out any of the subscripts $\ps{P},\secprop,\rel$ if they can be inferred from context.
The output of the apv is called a \emph{verdict}.
\begin{example}
  \label{ex:db-apv}
  In the situation of \cref{ex:db}, the following verdicts may be returned by the apv:
  \begin{labeling}[]{$\verdict{(M), (E_1, E_2)}$}
    \item[$\mkern2mu\verdict{}$]
      The empty verdict---no violation and no parties to blame.
    \item[$\verdict{(M)}$]
      The manager leaked the data on its own.
    \item[$\verdict{(E_1, E_2)}$]
      The employees colluded to leak the data.
    \item[$\verdict{(M), (E_1, E_2)}$]
      The manager as well as the employees leaked the data. \qedhere
  \end{labeling}
\end{example}

Each set $S \in \apv_{\ps{P},\secprop,\rel}(t)$ satisfies \hyperref[it:sc-1]{SC1}, \hyperref[it:sc-2]{SC2}, and \hyperref[it:sc-3]{SC3}.
\Cref{eq:apv-req-1} ensures that a violation occurred and therefore at least one party in $S$ deviated in $t$.
If not all parties in $S$ would deviate, there would be a counterfactual trace $t''$, where a strict subset of $S$ would deviate, thereby violating \cref{eq:apv-req-3}.
Hence, \hyperref[it:sc-1]{SC1} is satisfied.
\hyperref[it:sc-2]{SC2} is captured by \cref{eq:apv-req-2}, which ensures that there exists a counterfactual trace $t'$, showing that the parties in $S$ are sufficient to cause a violation.
\hyperref[it:sc-3]{SC3} follows directly from \cref{eq:apv-req-3}.

The following corollary shows that accountability with respect to $\secprop$ implies verifiability of $\secprop$.
If no violation occurred, no parties are blamed.
If no parties are blamed, no violation occurred.
\begin{corollary}
  \label{cor:apv}
  For all traces $t$,
  $  \apv_{\ps{P},\secprop,\rel}(t) = \verdict{} \iff t \holds \secprop$.
\end{corollary}
\begin{proof}
  Assume $t \holds \secprop$.
  Then $\apv_{\ps{P},\secprop,\rel}(t) = \verdict{}$ follows by \cref{def:apv}.
  For the other direction,
  assume $t \holds \neg \secprop$.
  As $\rel$ is reflexive, \cref{eq:apv-req-2} holds for $t$ and $S = \cor(t)$.
  If $S$ is already minimal, there does not exist a trace $t''$ which corrupts a strict subset of $S$ and thus $\apv_{\ps{P},\secprop,\rel}(t) = \Set{ S } \neq \verdict{}$.
  If $S$ is not minimal, there exists a trace $t''$ which corrupts $S' \subsetneq S$.
  The counterfactual trace $t'$ can then be instantiated with $t''$ and $S'$.
  If $S'$ is not minimal, this step can be repeated until a minimal set $S\sy{3}$ is obtained.
  As the cardinality of the sets decreases in each step, this approach is guaranteed to terminate.
\end{proof}

The apv can only be computed after the fact, that is, it requires full knowledge of the actual trace $t$.
The task of an accountability protocol is to always compute the apv without this information.
For generality, we assume that an accountability protocol comes with
a total function that extracts this information, the \emph{verdict function}:
\begin{equation}
  \label{eq:vf-dec}
  \vf(t) : \traces(\ps{P}) \to 2^{2^{\Parties}}\,.
\end{equation}

Accountability is now defined in terms of the apv and the verdict function.
If the apv coincides for all traces with the verdict function, the latter provides the protocol with accountability for a security property $\secprop$.
\begin{definition}[Accountability]
  \label{def:acc}
  A verdict function $\vfup$ provides a protocol $\ps{P}$ with accountability for a security property $\secprop$ with respect to a relation $\rel$, if
  \begin{equation}
    \Forall{ t \in \traces(\ps{P}) } \vf(t) = \apv_{\ps{P},\secprop,\rel}(t)\,. \tag{$\Acc_{\ps{P},\secprop,\rel}^{\vf}$}
  \end{equation}
\end{definition}

\begin{remark}[Counterfactual relation]
  \textcite{kuennemann2019automated} note that there is no consensus
  in the causality literature about how actual and counterfactual
  scenarios should relate.  They propose three approaches for relating
  actual and counterfactual traces: By control flow, by kind of
  violation, and the weakest relation with respect to
  \cref{eq:rel-corrupted}.
  Our focus will be on relating traces with the same kind of
  violations. As we will discuss in \cref{sec:count-rel}, our method
  may also be used to encode other relations. The axiomatic
  characterization in the next section, however, is independent of
  the choice of $\rel$.
\end{remark}

\section{Axiomatic characterization}
\label{sec:verif-acc}

Accountability (\cref{def:acc}) requires that the apv coincides with
a given verdict function for all traces of the protocol.
\begin{full}
Since the apv can only be computed after the fact and the number of
traces is most often infinite, this coincidence cannot be shown
directly.
However, the definition of the apv imposes multiple requirements on the verdicts returned by the verdict function.
\end{full}
In this \lcnamecref{sec:verif-acc},
we reformulate this requirement 
into five equivalent \emph{verification conditions}
that are sound and complete.
Soundness allows us to prove that a verdict function provides
a protocol with accountability by verifying that all conditions hold.
Completeness ensures that if a verdict function provides a protocol
with accountability, then all conditions hold.

This axiomatic characterization  bears resemblance to the verification
conditions  presented by \textcite{kuennemann2019automated}, but is
more general. It is valid for any counterfactual relation and any
verdict function. It is also simpler.%
\footnote{%
    More precisely, the present conditions (P) differ from the
    coarse-grained conditions (C)~\cite[Section~III]{kuennemann2019automated}
    and the fine-grained condition (F)~\cite[Section~IV]{kuennemann2019automated}
    as follows. 
    The completeness condition in (F) was found to be incompatible with the definition of the apv.
    Completeness in (P) is necessarily weaker and in line with the requirements of the other conditions (sufficiency, minimality, uniqueness).
    Verification is the same in all three.
    Sufficiency in (P) is slightly weaker than sufficiency in (C) and (F),
    as it allows for the witness trace to corrupt a subset of the
    blamed parties.
    Additionally, sufficiency in (P) requires no violation, but this
    requirement is
    superfluous, as it follows from verifiability.
    Sufficiency for composite verdicts in (F) follows from sufficiency in
    (P).
    Uniqueness in (P) and uniqueness for singletons in (F) are the same.
    Uniqueness in (C) is logically equivalent, but expressed differently.
    Minimality in (P) is weaker than in (C).  Minimality for composite
    verdicts in (F) can be considered equivalent, but for singleton
    verdicts it vanishes, because it follows from uniqueness in (F).
}\begin{full}\enlargethispage{2\baselineskip}\end{full}
These axioms will help us derive verification
conditions for unbounded sets of participants in the next section in
a systematic manner.

A verdict function $\vf$ providing accountability for $\varphi$
and $\rel$ is characterized by the following axioms.

\begin{conditions}
  \item[Verifiability (\cndlabel{\ensuremath{\Ver_{\secprop}}}{cnd:ver})]
    This follows directly from \cref{cor:apv}.
    \begin{equation*}
      \Forall{ t } \vf(t) = \verdict{} \iff t \holds \secprop
    \end{equation*}
    We require that whenever the verdict function returns an empty verdict, the security property holds.
  \item[Minimality (\cndlabel{\ensuremath{\Min_S}}{cnd:min})] 
    This follows directly from \cref{cor:min-apv}. 
    \begin{equation*}
      \Forall{ t } S \in \vf(t) \implies \NExists{ \qs{S'} } S' \in \vf(t) \land S' \subsetneq S
    \end{equation*}
    We require that the verdict does not contain a strict subset of one of its sets.
    Intuitively, this axiom ensures that we only blame the least number of parties which caused a violation.
  \item[Sufficiency (\cndlabel{\ensuremath{\Suff_S}}{cnd:suff})]
    This axiom is similar to \cref{eq:apv-req-2} and guarantees that each set of parties in a verdict is sufficient to cause a violation on their own.
    \begin{align*}
        \Forall{ t } S \in \vf(t) \implies \begin{aligned}[t]
          \Exists{ \qs{t'} } &\vf(t') = \verdict{ S } \\
                   {}\land{} &\cor(t') \subseteq S \land \rel(t,t')
          \end{aligned}
    \end{align*}
     For each set of parties in a verdict, there exists a related trace in which only a subset of these parties has been corrupted and for which the verdict function returns a singleton verdict only blaming these parties.
\end{conditions}

We could also define a slightly stronger sufficiency condition, where we would require equality between the set of corrupted parties and the set in the verdict, that is, $\cor(t') = S$.
Instead, we capture this requirement in its own condition---uniqueness.%
\footnote{
  The name goes back to the case distinction used in the verdict function of \textcite{kuennemann2019automated}, in which the condition ensures that only a unique, sufficient, and minimal verdict exists for each case.
}
With this approach, we get more precise information when
a condition does not hold.
\begin{conditions}
  \item[Uniqueness (\cndlabel{\ensuremath{\Uniq_S}}{cnd:uniq})]
  This condition guarantees that all parties in a verdict have been corrupted; or in other words, no honest parties are blamed for a security violation.
  \begin{equation*}
    \Forall{ t } S \in \vf(t) \implies S \subseteq \cor(t)
  \end{equation*}
\end{conditions}

The previous four conditions state the requirements that a group of parties in the verdict must satisfy.
The next condition ensures that indeed all groups of parties which satisfy these requirements are included in the verdict.

\begin{conditions}
  \item[Completeness (\cndlabel{\ensuremath{\Comp_S}}{cnd:comp})]
  \begin{align*}
    \Forall{ t } &\Bigg[ \begin{aligned}
          \Exists{ \qs{t'} } &\vf(t') = \verdict{S} \\
                   {}\land{} &\cor(t') \subseteq S \land \rel(t,t')
       \end{aligned} \Bigg] \conjlabel{eq:conj-comp-1}\\
       {}\land{} &\del[\big]( \NExists{ \qs{S'} } S' \in \vf(t) \land S' \subsetneq S ) \conjlabel{eq:conj-comp-2}\\
       {}\land{} &S \subseteq \cor(t) \conjlabel{eq:conj-comp-3}\\
       {}\land{} &t \holds \neg \secprop \conjlabel{eq:conj-comp-4}\\
       \phantom{\land} &\implies{} S \in \vf(t)
  \end{align*}
\end{conditions}

We write $\Suff$, $\Ver_{\secprop}$, $\Min$, $\Uniq$, and $\Comp$ if the respective condition holds for all $S$.
We denote the conjunction of these conditions by $\VC_{\secprop}$ or by $\VC$ if the security property can be inferred from context.

\begin{full}
\begin{example}
  \label{ex:db-verif-cond}
  Consider the protocol described in \cref{ex:db} with the following verdict function.
  \begin{equation}
    \label{eq:ex-db-verif-cond-vf}
    \vf(t) = \begin{dcases}
      \verdict{(\Party{M})}                & \hspace{-4pt}\text{if } \fact{Access} \del[\big]( \Party{M} )@i \in t \\
      \verdict{(\Party{E_1}, \Party{E_2})} & \hspace{-4pt}\text{if } \fact{Access} \del[\big]( \del< \Party{E_1}, \Party{E_2} > )@i \in t \\
      \mkern2mu\verdict{}                  & \hspace{-4pt}\text{otherwise}
     \end{dcases}
  \end{equation}

  To prove that this verdict function provides accountability for not leaking data, we have to verify that the verdict function is indeed total and that all verification conditions hold.
  The former follows directly from \cref{eq:ex-db-verif-cond-vf}.
  For the latter, one would show the following.
  \begin{labeling}{\nameref{cnd:suff}:}
    \item[\nameref{cnd:suff}:] Knowing the signing key of the parties in the verdict is sufficient to leak the data.
    \item[\nameref{cnd:ver}:] If no party accesses the data (the otherwise case in \cref{eq:ex-db-verif-cond-vf}), no data can be leaked.
    \item[\nameref{cnd:min}:] Without any signing key or only the signing key of a single employee, the data cannot be leaked.
    \item[\nameref{cnd:uniq}:] Accessing and leaking the data requires corrupting the respective parties.
    \item[\nameref{cnd:comp}:] Each set of parties satisfying the above conditions is included in the verdict. \qedhere
  \end{labeling}
\end{example}
\end{full}

\begin{theorem}
  \label{thm:axiom-eq}
  For any protocol $\ps{P}$, security property $\secprop$, and verdict function $\vfup$, 
  $\vfup$ provides $\ps{P}$ with accountability for $\secprop$
  iff
  $\VC$.
\end{theorem}
\begin{proof}
    In
    \cref{sec:verif-acc-proof},
    we show soundness and
    completeness in two separate theorems.
\end{proof}

\section{Verdict Functions for\texorpdfstring{\\}{ }Unbounded Sets of Participants}
\label{sec:verdict-functions}

The structure of verdict functions proposed by
\textcite{kuennemann2019automated}
\begin{full}
and exemplified in
\cref{ex:db-verif-cond}
\end{full}
considers an explicit mapping from observations, i.e.,
sets of traces, to
 verdicts.
All parties that can occur in a verdict are thus fixed a priori.
This prohibits the analysis of several protocol instances in parallel
and is inadequate for protocols such as TLS, where a single
responder may react to incoming requests from many clients.

If we allow for multiple protocol sessions and consider the set of
parties that participate to be unbounded, then, for some protocols,
we cannot bound the number of possible verdicts or their size.
Therefore, we must define the verdict function indirectly.
To this end, we lift the \emph{accountability tests} of \textcite{bruni2017automated},
which determine whether a given party is to blame,
to \emph{case tests}, which can contain variables instead of concrete parties.
Case tests are trace properties with free variables. Each free
variable is instantiated with a party that should be blamed for
a violation. A case test ought to have at least one free variable
and there should be at least one trace where it applies.

\begin{definition}[Case test]
    A case test $\ct$ is a trace property which satisfies
    \begin{enumerate}[label=(\alph*)]
      \item $\abs{ \fv(\ct) } \geqslant 1$ and
      \item $\Exists{ t, \inst } t \holds \ct \inst$\,.
    \end{enumerate}
\end{definition}
We say a case test $\ct$ \emph{matches} a trace $t$ if there exists an instantiation $\inst$ such that $t \holds \ct \inst$.
\begin{full}
We say a case test $\ct$ matches if the trace can be inferred from context.
\end{full}
The verdict function is now given as the union of all matches.
\begin{definition}[Verdict function]
  \label{def:verdict-function}
  Let $\CT$ be a set of case tests.
  The verdict function induced by $\CT$ is given by
  \begin{equation}
    \vf_{\CT}(t) \defas \Union{ \ct \in \CT } \DisplaySet[\Big]{ \fv(\ct) \inst \given \Exists{ \inst } t \holds \ct \inst }\,,
  \end{equation}
  where the union is modulo the equational theory $E$.
\end{definition}
In the following, we assume a fixed set of user-defined case tests which are denoted by $\CT = \ct_1, \dots, \ct_n$ and write $\vf(t)$.

\begin{example}
  \label{ex:db-case-tests}
  Consider the protocol described in \cref{ex:db} in the multi-session setting, that is, there may be multiple managers, employees, and data leaks.
  There are two possibilities, how a data leak can arise.
  Either by a manager or by two colluding employees.
  We want to hold all groups of parties accountable which are responsible for a leak.
  In contrast to the single-session setting, the protocol must now provide evidence which group of parties leaked the data.
  Only knowing the parties which accessed the data is not sufficient to identify the parties responsible for a violation.
  In the case of a single violation, we would suspect all groups of parties that accessed the data.
  
  The security property indicates that neither a manager nor employees leaked the data.
  \begin{align*}
    \secprop \defas \NExists{ m, e_i, e_j, \var{data}, i } &\fact{LeakManager}(\var{m}, \var{data})@i \\
                                                  {}\lor{} &\fact{LeakEmployees}(\var{e_i}, \var{e_j}, \var{data})@i
  \end{align*}

  We define the following two case tests.
  \begin{align*}
    \ct_1 &\defas \Exists{ \var{data}, i } \fact{LeakManager}(\var{m}, \var{data})@i \\
    \ct_2 &\defas \Exists{ \var{data}, i } \fact{LeakEmployees}(\var{e_i}, \var{e_j}, \var{data})@i
  \end{align*}

  We note that the identities of the manager ($m$) and employees ($e_i,e_j$) are free in $\ct_1$ and $\ct_2$ respectively.
  Given a trace $t$ where two managers $\Party{M_1}$, $\Party{M_2}$ and each pair of employees $\Party{E_1}$, $\Party{E_2}$, $\Party{E_3}$ caused a violation, the following instantiations exist for $\ct_1$ and $\ct_2$.
  \begin{equation*}
    \begin{aligned}[c]
      \inst_1^{(1)} &= \subst{ m \mapsto \Party{M_1} } \\
      \inst_1^{(2)} &= \subst{ m \mapsto \Party{M_2} }
    \end{aligned}
    \qquad
    \begin{aligned}[c]
      \inst_2^{(1)} &= \subst{ e_i \mapsto \Party{E_1}, e_j \mapsto \Party{E_2} } \\
      \inst_2^{(2)} &= \subst{ e_i \mapsto \Party{E_2}, e_j \mapsto \Party{E_3} } \\
      \inst_2^{(3)} &= \subst{ e_i \mapsto \Party{E_1}, e_j \mapsto \Party{E_3} }
    \end{aligned}
  \end{equation*}

  We obtain all singleton verdicts of the trace by applying the instantiations to the free variables of the case tests.
  Hence, the complete verdict is
  \begin{equation*}
    \vf(t) = \verdict{(\Party{M_1}), (\Party{M_2}), (\Party{E_1}, \Party{E_2}), (\Party{E_2}, \Party{E_3}), (\Party{E_1}, \Party{E_3})}\,. \qedhere
  \end{equation*}
\end{example}

This example also illustrates that case tests are well suited to
distinguish different kinds of a violation, which are identified by the
test and its instantiation. We can formalize this notion by assigning each trace the set of case tests with their
corresponding satisfying instantiations.
\begin{equation}
  \label{eq:ctr}
  \ctr(t) \defas \Union{ i \in [n] } \Set[\big]{ (\ct_i, \inst) \given \Exists{ \inst } t \holds \ct_i \inst }
\end{equation}
We call a trace $t$ \emph{single-matched} if $\abs{ \ctr(t) } = 1 $
and \emph{multi-matched} if $\abs{ \ctr(t) } > 1 $.
\begin{example}
  \label{ex:db-ctr}
  In the situation of \cref{ex:db-case-tests}, we obtain
  \begin{equation*}
    \ctr(t) = \Set[\big]{ (\ct_1, \inst_1^{(1)}), (\ct_1, \inst_1^{(2)}), (\ct_2, \inst_2^{(1)}), (\ct_2, \inst_2^{(2)}), (\ct_2, \inst_2^{(3)}) }\,. \qedhere
  \end{equation*}
\end{example}
We use the following corollary to justify switching between the verdict-based notation of \cref{sec:verif-acc} and the notation based on case tests of this \lcnamecref{sec:verdict-functions}.
\begin{corollary}
  \label{cor:verdict-function}
  \cref{def:verdict-function} implies that
  for all traces $t$
  \begin{equation*}
    S \in \vf(t) \iff \Exists{ i, \inst } t \holds \ct_i \inst \land \fv(\ct_i) \inst = S\,.
  \end{equation*}
\end{corollary}
We see that $\ctr(t)$ contains all the information to compute the
verdict for the trace $t$.  \cref{def:verdict-function} implies
\begin{equation}
  \label{eq:ctr-vf}
  \ctr(t') \subseteq \ctr(t) \implies \vf(t') \subseteq \vf(t)\,.
\end{equation}
However, $\ctr$ provides a more precise picture, since the same set in the
verdict may be produced by multiple case tests and instantiations.

We can now instantiate the verification conditions from \cref{sec:verif-acc}
with case tests. If $\CT$ is finite, we obtain a finite set of
conditions, all of which (except sufficiency) are predicates on traces,
but not yet trace formulas according to \cref{def:trace-formula}.
We first apply \cref{cor:verdict-function} to each occurrence of $S \in \vf(t)$ and $\vf(t) = \verdict{ S }$ in the conditions.
Since the original conditions are parameterized by $S$, the resulting conditions are parameterized by a case test $\ct_i$ and an instantiation $\inst$.
We reparameterize these conditions with a case test $\ct_i$ by introducing quantifiers for the instantiations.
As the set of case tests is finite, we also replace quantification over case tests by conjunctions and disjunctions.
For an instantiation $\inst$, we have
\begin{align}
  \begin{split}
    \Exists{ i } t \holds \ct_i \inst \iff \Or*{ i \in [n] } t \holds \ct_i \inst\,, \label{eq:ct-ex-quan}
  \end{split}\\
  \shortintertext{and}
  \begin{split}
    \Forall{ i } t \holds \ct_i \inst \iff \And*{ i \in [n] } t \holds \ct_i \inst\,. \label{eq:ct-all-quan}
  \end{split}
\end{align}

Moreover, we split the equivalence in the verifiability condition \nameref{cnd:ver}.
This step is not a technical requirement, but we may gain more insight in case a condition does not hold.
Finally, we obtain the following intermediate representation.
\begin{conditions}
  \item[Sufficiency (\cndlabel{\ensuremath{\Suff_{\ct_i}^{\mathrm{in}}}}{cnd:suffin})]
    Assume a case test matches a trace $t$.
    Then there exists a related trace $t'$ in which only the instantiated parties are corrupted.
    Moreover, if multiple case tests match, all sets of instantiated parties have to be the same.
    This ensures that the verdict is a singleton.
    \begin{align*}
        \Forall{ t, \inst } &t \holds \ct_i \inst \implies \\
        \Exists{ \qs{t'} } &\del[\Big][ \And[r]{ j \in [n] } \Forall{ \qs{\inst'} } t' \holds \ct_j \inst' \implies \fv(\ct_i) \inst = \fv(\ct_j) \inst' ] \\
        &\land \cor(t') \subseteq \fv(\ct_i) \inst \\
        &\land \rel(t,t')
    \end{align*}
  \item[Verifiability Empty (\cndlabel{\ensuremath{\VerE_{\secprop}^{\mathrm{in}}}}{cnd:verin-e})]
    If there is no case test that matches, then the security property holds.
    This ensures that the security property can only be violated in the ways described by the case tests.
    \begin{align*}
      \Forall{ t } \del[\Big][ \And[r]{ i \in [n] } \NExists{ \inst } t \holds \ct_i \inst ] \implies t \holds \secprop
    \end{align*}
  \item[Verifiability Nonempty (\cndlabel{\ensuremath{\VerNE_{\secprop,\ct_i}^{\mathrm{in}}}}{cnd:verin-ne})]
    The condition requires that if a case test matches, then the security property is violated.
    This ensures that each case test describes a way to violate the security property.
    \begin{align*}
      \Forall{ t, \inst } t \holds \ct_i \inst \implies t \holds \neg \secprop
    \end{align*}
  \item[Minimality (\cndlabel{\ensuremath{\Min_{\ct_i}^{\mathrm{in}}}}{cnd:minin})]
    The condition ensures that, when a case test matches, then no other case test matches with a strict subset of the instantiated parties.
    \begin{align*}
      \Forall{ t, \inst } t \holds \ct_i \inst \implies \And*{j \in [n] } \NExists{ \qs{\inst'} } t \holds \ct_j \inst' \land \fv(\ct_j) \inst' \subsetneq \fv(\ct_i) \inst
    \end{align*}
  \item[Uniqueness (\cndlabel{\ensuremath{\Uniq_{\ct_i}^{\mathrm{in}}}}{cnd:uniqin})]
    The condition requires that the instantiated parties of a case test have been corrupted.
    This ensures that we do not blame honest parties for a security violation.
    \begin{align*}
      \Forall{ t, \inst } t \holds \ct_i \inst \implies \fv(\ct_i) \inst \subseteq \cor(t)
    \end{align*}
\end{conditions}

\begin{remark}
  The completeness condition \nameref{cnd:comp} does not need to be encoded as a trace property.
  We show in \cref{lem:comp} that it follows from \nameref{cnd:verin-ne}, \cref{eq:ctr-vf}, and \nameref{cnd:rel-el}, a requirement on the counterfactual relation we introduce in the next \lcnamecref{sec:count-rel}.
\end{remark}

\section{Counterfactual Relation}
\label{sec:count-rel}

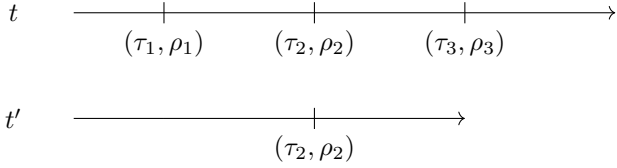
\begin{figure}
\centering
\caption{Example: We consider $t'$ a valid counterfactual for $t$.}
\label{fig:ex-counterfactual}
\begin{tikzpicture}[scale=2
]

\draw (0,0) node {$t$};
\draw (0.4,0) -- (1,0);
\draw (1,0) -- (2,0);
\draw (2,0) -- (3,0);
\draw[->] (3,0) -- (4,0);

\foreach \x in {1,2,3}
   \draw (\x cm,2pt) -- (\x cm,-2pt);

   \draw (1,0) node[below=3pt] {$ (\tau_1,\rho_1) $} node[above=3pt] {};
\draw (2,0) node[below=3pt] {$ (\tau_2,\rho_2) $} node[above=3pt] {};
\draw (3,0) node[below=3pt] {$ (\tau_3,\rho_3) $} node[above=3pt] {};

\draw (0,-2em) node {$t'$};
\draw (0.4,-2em) -- (1,-2em);
\draw (1,-2em) -- (2,-2em);
\draw[->] (2,-2em) -- (3,-2em);

\foreach \x in {2}
   \draw (\x cm,-2em+2pt) -- (\x cm,-2em-2pt);

\draw (2,-2em) node[below=3pt] {$ (\tau_2,\rho_2) $} node[above=3pt] {};

\end{tikzpicture}

\end{figure} 

As we consider an unbounded number of sessions, we can, in general,
expect to have multiple causally independent security violations in the
same trace.
Consider $t$ in \cref{fig:ex-counterfactual}, where three case
tests $\tau_1,\ldots,\tau_3$ match. By
\hyperref[cnd:vertp-ne]{$\VerNE_{\secprop}$},
each match implies a security violation by itself.
For our counterfactual analysis, we want to consider traces that
contain only one of these matches causally relevant.
We require the counterfactual relation to be compatible with this
intuition, which we formalize as follows.
Note that neither condition restricts the relation for non-violating
traces.
Indeed, the relation is irrelevant for the apv of those (see
\cref{def:apv}).

 \begin{conditions}
   \item[Relation Introduction (\cndlabel{\ensuremath{\RelI}}{cnd:rel-in})]
       For all $i \in [n]$,  traces $t,t'$ and instantiations $\inst$
       \begin{align*}
         \phantom{{}\land{}} &t \holds \ct_i \inst \\
         {}\land{} &\ctr(t') = \Set[\big]{ (\ct_i, \inst) } \\
         {}\land{} &\cor(t') = \fv(\ct_i) \inst \subseteq \cor(t) \\
         &\implies \rel(t,t')\,.
       \end{align*}
 
     If there is at least one match in some trace $t$ and we can
     identify a trace $t'$ with the exact same match,
     corrupting no more parties than $t$ and exactly those indicated
     by the match, then
     we consider $t'$ a relevant counterfactual.

  \item[Relation Elimination (\cndlabel{\ensuremath{\RelE}}{cnd:rel-el})]
    For all traces $t$, $t'$
    \begin{equation*}
      \rel(t,t') 
      \land t \holds \neg\secprop
      \land t' \holds \neg\secprop
      \implies \ctr(t') \subseteq \ctr(t) 
    \end{equation*}

    Intuitively, if $t'$ is a relevant counterfactual for $t$, then it
    cannot have \emph{additional}
    matches.
\end{conditions}

\textcite{kuennemann2019automated} discuss the lack of a consensus in
the causality literature about how actual and counterfactual scenarios
should relate.
They consider three frequently used relations:
$\rel_k$, 
where two traces relate if
they have a \enquote{similar kind} of violation;
$\rel_c$,
where they need to
have the same control flow;
 and
$\rel_w$, 
which is the
weakest possible relation, where $\rel_w(t,t') \iff \cor(t') \subseteq
    \cor(t)$.
Neither gives an indication of how to deal with several parallel infractions in
the same trace, but they give us a framework to discuss the present
proposal.

Considering the kind of violation, $\rel_k$ provides the most
promising interpretation of case tests. 
This notion originates from criminal law and is used to solve 
causal problems with the classical \enquote{what-if} by considering the event
in question in greater detail, e.g.,
by distinguishing death from shooting from death from poisoning
(\cite[p.~188]{dressler1995understanding}; see
also~\cite[p.~46]{Mackie1980-MACTCO-17}). 
As such, this notion is informal and depends on intuition. 
Using case tests, we could 
formalize $\rel_k$ using $\rel_\ctr$ with
    $\rel_{\ctr}(t,t')$ iff
\[
        \emptyset \neq \ctr(t') \subseteq \ctr(t)
        \land \cor(t') \subseteq \cor(t)\,.
\] 
The counterfactual has a subset of the matches of the actual, but at
least one.
$\rel_{\ctr}$ is consistent with both
\hyperref[cnd:rel-el]{$\RelE$}
and
\hyperref[cnd:rel-in]{$\RelI$}.

While $\rel_k$ is only informally defined, it is usually straightforward
to apply it to a given protocol and a set of case tests.
Case by case, we can thus
confirm
that $\rel_{\ctr}$ is a formalization of $\rel_k$.
Each test and possible instance should mark a different \enquote{kind} of
violation.

By contrast,
a direct adoption of the control flow aware relation
$\rel_c$~\parencite{datta2015program,kuennemann2018accountability,kuntz2011probabilistic}
would not allow for holding all involved parties responsible, as the
control flow of $t'$ (\cref{fig:ex-counterfactual}) is clearly different from $t$. 
If we relax the relation to accept
a counterfactual trace if its control flow is a prefix of the
actual control flow, we would  only collect the parties involved in the
\emph{first} violation, which is not our goal.
We should thus consider only the control flow per session,
and allow the order of sessions to be changed.\footnote{%
As our execution model cannot capture the control flow of deviating
parties, this only concerns the trusted parties.
}
This could be encoded by
splitting the case tests, so that each test applies only to a single
per-session control flow.

The weakest possible counterfactual relation $\rel_w$
can not guarantee
\hyperref[cnd:rel-el]{$\RelE$}
as a related trace could match
a completely different case test.

\begin{full}
\begin{example}
    In general, it is possible that two traces have the same verdict,
    but are not related according to
\hyperref[cnd:rel-el]{$\RelE$}.
  Consider two traces $t$, $t'$ such that $\ctr(t') = \Set[\big]{ (\ct_i, \inst_i) }$ and $\ctr(t') = \Set[\big]{ (\ct_j, \inst_j) }$ with $\fv(\ct_i) \inst_i = \fv(\ct_j) \inst_j$.
  Then $\vf(t') \subseteq \vf(t)$ but $\ctr(t') \nsubseteq \ctr(t)$.
  In both traces, the same set of parties causes a violation, but in $t'$ in the way described by $\ct_i$ and in $t$ in the way described by $\ct_j$.
\end{example}
In the soundness proof, we need \nameref{cnd:rel-in} to introduce the relation which occurs in \nameref{cnd:suff} but not in the trace properties.
In the completeness proof, we need \nameref{cnd:rel-el} to lift the verdict-based verification conditions to the ones based on case tests.
In this way, the relation has to provide the expressiveness missing in the verdict-based conditions.
\end{full}

\section{ Verification Conditions as Trace Properties }
\label{sec:verification-conditions-trace-prop}

In this section, we bring
the axioms from \cref{sec:verif-acc},
which we instantiated with
case tests in the last section,
into a form that can be verified in an automated way.
The most challenging among these is sufficiency.

\subsection{Subset relations}
\label{sec:towards-trace-props}

To encode minimality, sufficiency, and uniqueness, we need to express
the subset operator and the function $\cor$ in terms of protocol actions.
Let $\vec{a} = ( a_1, \dots, a_m )$ and $\vec{b} = ( b_1, \dots, b_n )$. 
The strict subset operator in \nameref{cnd:minin} can be expressed by
\begin{equation}
  \label{eq:strict-sub-tp}
  \TraceProp[\big]{ \vec{a} \subsetneq \vec{b} } \defas \del[\Big][ \And{ i \in [m] } \Or[r]{ j \in [n] } a_i = b_j ] \land \del[\Big][  \Or{ j \in [n] } \And[r]{ i \in [m] } b_j \neq a_i ]\,.
\end{equation}

The corruption of a party $\Party{A}$ is recorded as an action $\fact{Corrupted}(\Party{A})@i$ in the trace.
Hence, the subsets in \nameref{cnd:suffin} and \nameref{cnd:uniqin} can be expressed for a trace $t$ by
\begin{align}
\cor(t) \subseteq \vec{a} & \iff t \holds \TraceProp[\big]{ \fact{Corrupted} \subseteq \vec{a} } \label{eq:cor-sub-tp}\\
    \vec{a} \subseteq \cor(t) & \iff t \holds \TraceProp[\big]{ \vec{a} \subseteq \fact{Corrupted} } \label{eq:cor-sup-tp}
\end{align}
where
\begin{align*}
    \TraceProp[\big]{ \fact{Corrupted} \subseteq \vec{a} } &\defas \Forall{ \var{x}, k } \fact{Corrupted}(\var{x})@k \impliestp \Or*{ i \in [m] } \var{x} = \var{a_i} \\
    \TraceProp[\big]{ \vec{a} \subseteq \fact{Corrupted} } &\defas \And*{ i \in [m] } \Exists{ k } \fact{Corrupted}(\var{a_i})@k\,.
\end{align*}

\subsection{Sufficiency as a trace property}

These encodings allow us to express all conditions as trace properties, except one: \nameref{cnd:suffin}.
It has two particularities.
First, it is of the form $\Forall{ t } \Exists{  \qs{t'} } \gamma(t, t')$, which classifies it as a
hyperproperty~\parencite{clarkson2010hyperproperties}.
Since hyperproperties are in general more expressive than trace properties,
they cannot be directly converted to the latter.
Second, it is the only condition that contains the counterfactual relation $\rel$.

To derive a trace property, we need to get rid of the outermost
universal quantifier and abstract the relation $\rel$.
To avoid the quantifier, we will focus on single-matched
traces, i.e., traces with exactly one violation and introduce three
additional conditions. They ensure
that there exists a single-matched trace
\begin{enumerate*}[label=(\alph*)]
\item for any case test,
\item for any instance thereof, and
\item that matching assignments are always injective.
\end{enumerate*}
These properties can be considered well-formedness
conditions on the case tests and are automatically verified.
They define a class of protocols and case tests for which our trace
properties are sound and complete.\footnote{%
Alternatively, they can be understood as part of the verification
conditions. In this case, we offer two sets of 
conditions, one that is sound and one that is complete.
}
To abstract the relation $\rel$, we make use of the assumption
introduced in \cref{sec:count-rel}, which is not automatically
verified. 

\begin{table*}[th]
    \newlength{\logcellwidth}
    \setlength{\logcellwidth}{4cm}

    \newcommand{\vertcent}[1]{\begin{tabular}{@{}c@{}}#1\end{tabular}}

    \centering
    \caption{Verification conditions}
    \label{tab:ver-cond}
    \small
    \begin{tabular}{llll}
        \toprule
        & name & definition & logical relation \\
        \midrule
        \multirow{3}{*}[-2ex]{\vertcent{\rotatebox[origin=c]{90}{sufficiency (tr. prop.)}}}
            & \cndlabel{\ensuremath{\SinM_{\ct_i}}}{cnd:single} & $\begin{aligned}
                  &\ps{P} \holds^{\exists}
                  \Exists{ \vec{v} } \ct_i\subst{\vec{v}} &\land \del[\Big][ \Forall{ \vec{w} } \ct_i \subst{\vec{w}} \impliestp \vec{w} = \vec{v} ] 
                  \land \del[\Big][ \And[r]{ j \in [n] \setminus \Set{i} } \NExists{ \vec{x} } \ct_j \subst{\vec{x}} ]
                  \end{aligned}$
                  & \multirow{3}{*}[-5ex]{\parbox[m]{\logcellwidth}{
                        \centering
                        $\hyperref[cnd:sufftp]{\Suff^{\tp}} \land \hyperref[cnd:vertp-ne]{\VerNE_{\secprop}^{\tp}} \land \hyperref[cnd:uniqtp]{\Uniq^{\tp}} \land \hyperref[cnd:ins-inj]{\InsI} \land \hyperref[cnd:rep-prop]{\RepP} \implies \hyperref[cnd:suff]{\Suff}$ \\
                        (\Cref{lem:suff-snd}) \\ \medskip
                        $\hyperref[cnd:suff]{\Suff} \land \hyperref[cnd:single]{\SinM} \implies \hyperref[cnd:sufftp]{\Suff^{\tp}}$ \\
                        (\Cref{lem:suff-cmpl})}} \\
            & \cndlabel{\ensuremath{\InsI_{\ct_i}}}{cnd:ins-inj} & $\begin{aligned}
                &\ps{P} \holds^{\forall}
                \Forall{ \vec{v} } \ct_i \subst{\vec{v}} \impliestp \And{ i \in \idx(\vec{v}) } \And[r]{ \substack{j \in \idx(\vec{v}) \\ j \neq i } } v_i \neq v_j
                \end{aligned}$ \\
            & \cndlabel{\ensuremath{\Suff_{\ct_i}^{\tp}}}{cnd:sufftp} & $\begin{aligned}
                  \ps{P} \holds^{\exists}
                  \Exists{ \vec{v} } \ct_i\subst{\vec{v}} &\land \del[\Big][ \Forall{ \vec{w} } \ct_i \subst{\vec{w}} \impliestp \vec{w} = \vec{v} ] 
                  \land \del[\Big][ \And[r]{ j \in [n] \setminus \Set{i} } \NExists{ \vec{x} } \ct_j \subst{\vec{x}} ] \\
                  &\land \Forall{ \var{a}, k } \fact{Corrupted}(\var{a})@k \impliestp \Or*{ \ell \in \idx(\vec{v}) } \var{a} = \vec{v_{\ell}}
                  \end{aligned}$ \\ \midrule
        \multirow{4}{*}[1ex]{\vertcent{\rotatebox[origin=c]{90}{other conditions (tr. prop.)}}}
            & \cndlabel{\ensuremath{\VerE_{\secprop}^{\tp}}}{cnd:vertp-e} & $\begin{aligned}
                  \ps{P} \holds^{\forall} 
                  \del[\Big][ \And[r]{ i \in [n] } \NExists{ \vec{v} } \ct_i \subst{\vec{v}} ] \impliestp \secprop
                  \end{aligned}$
                  & \multirow{2}{*}[-2ex]{\parbox[m]{\logcellwidth}{
                        \centering
                        $\hyperref[cnd:ver]{\Ver_{\secprop}} \iff \hyperref[cnd:vertp-e]{\VerE_{\secprop}^{\tp}} \land \hyperref[cnd:vertp-ne]{\VerNE_{\secprop}^{\tp}}$ \\
                        (\Cref{lem:ver-equiv})}} \\[3ex]
            & \cndlabel{\ensuremath{\VerNE_{\secprop,\ct_i}^{\tp}}}{cnd:vertp-ne} & $\begin{aligned}
                  \ps{P} \holds^{\forall} 
                  \Forall{ \vec{v} } \ct_i \subst{\vec{v}} \impliestp \neg \secprop
                  \end{aligned}$ & \\[3ex]
            & \cndlabel{\ensuremath{\Min_{\ct_i}^{\tp}}}{cnd:mintp} & $\begin{aligned}
                  \ps{P} \holds^{\forall}
                  \Forall{ \vec{v} } \ct_i \subst{\vec{v}} \impliestp \And*{ j \in [n] } \NExists{ \vec{w} } \ct_j \subst{\vec{w}} \land \TraceProp[\big]{ \vec{w} \subsetneq \vec{v} }
                  \end{aligned}$
                  & \parbox[m]{\logcellwidth}{
                        \centering
                        $\hyperref[cnd:min]{\Min} \iff \hyperref[cnd:mintp]{\Min^{\tp}}$ \\
                        (\Cref{lem:min-equiv})} \\[3ex]
            & \cndlabel{\ensuremath{\Uniq_{\ct_i}^{\tp}}}{cnd:uniqtp} & $\begin{aligned}
                  \ps{P} \holds^{\forall}
                  \Forall{ \vec{v} } \ct_i \subst{\vec{v}} \impliestp \And*{ \ell \in \idx(\vec{v}) } \Exists{ k } \fact{Corrupted}(\vec{v_{\ell}})@k
                  \end{aligned}$ 
                  & \parbox[m]{\logcellwidth}{
                       \centering
                       $\hyperref[cnd:uniq]{\Uniq} \iff \hyperref[cnd:uniqtp]{\Uniq^{\tp}}$ \\
                       (\Cref{lem:uniq-equiv})} \\ \midrule
        \addlinespace[1ex]
        \multirow{2}{*}{\vertcent{\rotatebox[origin=c]{90}{syntactic}}}
            & \cndlabel{\ensuremath{\RepP_{\ct_i}}}{cnd:rep-prop} & $\begin{aligned}
                  &\Forall{ t, t', \inst, \inst' } t \holds \ct_i \inst \land \ctr(t') = \Set[\big]{ (\ct_i, \inst') } \\
                  & \implies \Exists{ \qs{t''} } \ctr(t'') = \Set[\big]{ (\ct_i, \inst) } \land \cor(t'') = \cor(t')(\inst \circ {\inst'}^{-1})
                  \end{aligned}$
                  & \parbox[m]{\logcellwidth}{
                       \centering
                       $\hyperref[cnd:br]{\BR} \implies \hyperref[cnd:rep-prop]{\RepP}$} \\[4ex]
            & \cndlabel{\ensuremath{\BR}}{cnd:br} & $
                  \Forall { t, \sub \colon \Parties \leftrightarrow \Parties} t \in \traces(\ps{P}) \implies t \sub \in \traces(\ps{P})$
                  & \parbox[m]{\logcellwidth}{
                        \centering
                        $\Parties \subseteq \PN \land \fn(P) \cap \PN = \emptyset$\\
                    $\implies \hyperref[cnd:br]{\BR}$ \\
                    (\Cref{lem:br-sufficiency})} \\[1ex] \bottomrule
    \end{tabular}
\end{table*}

We can express that a trace is single-matched as a trace property $\SinM_{\ct_i}$ (see \cref{tab:ver-cond}).
For a trace $t$ to be single-matched, i.e., $\ctr(t) = \Set[\big]{
(\ct_i, \inst) }$, three conditions have to be satisfied.  First,
$\ct_i$ has to match $t$; second, if it matches multiple times, then
all variable assignments have to be equal; and finally, no other case
test may match $t$.
We write $\SinM$ if $\SinM_{\ct_i}$ holds for all $i \in [n]$.

With \nameref{cnd:single},  \cref{eq:cor-sub-tp},
and \nameref{cnd:rel-in},
we can express the consequent of \nameref{cnd:suffin} 
as a trace property
\begin{equation}
  \label{eq:suffin-pre}
  \Exists{ t, \inst } \ctr(t) = \Set[\big]{ (\ct_i, \inst) } \land \cor(t) \subseteq \fv(\ct_i) \inst
\end{equation}
This guarantees the existence of a single trace for
each case test, but not for all possible instantiations.
We thus need to ensure that if \cref{eq:suffin-pre} holds for
a single instantiation, then it also holds for all possible
instantiations.
To achieve this, we first introduce an additional requirement on the
counterfactual relation, the replacement property \nameref{cnd:rep-prop} (see \cref{tab:ver-cond}).

Assuming \nameref{cnd:single} holds, i.e., 
some single-matched $t'$ exists, then
\nameref{cnd:rep-prop} ensures that for any
multi-matched 
trace $t$ with a match for $\tau_i$ and $\inst$,
there is a single-matched trace $t''$ with the same case test and
instantiation as $t$. The property is slightly more general, as
$t'$ could corrupt more parties than
necessary. To illustrate this point: If $t'$ corrupts the minimal
set of parties, i.e., $\cor(t')=\fv(\tau_i)\rho'$, then 
$\cor(t'')=\fv(\tau_i)\rho$. 
We write $\RepP$ if $\RepP_{\ct_i}$ holds for all $i \in [n]$.
In other words, \hyperref[cnd:single]{$\SinM$} and
\hyperref[cnd:rep-prop]{$\RepP$} ensure that there is a decomposition
of each trace that separates interleaving causally relevant events so
they can be regarded in isolation.
A sufficient criterion is that traces are closed under bijective
renaming, which we can ensure syntactically by verifying that no
public names appear in the protocol and that $\fact{Corrupted}$ actions
contain only variables of sort public.\footnote{%
This simple check applies to both multiset-rewrite rules and SAPiC processes.
A more refined syntactic condition is possible by allowing for public names unless
they are compared to variables that occur in verdicts.
In our case studies, we verified this condition by hand.
}

To ensure that ${\inst'}^{-1}$ is well defined, we require 
each free variable to be instantiated with a distinct value.
This can be expressed
as a trace property, \hyperref[cnd:ins-inj]{Instance Injectivity $\InsI$} (see \cref{tab:ver-cond}).
We write $\InsI$ if $\InsI_{\ct_i}$ holds for all $i \in [n]$.
This condition is w.l.o.g.  If a case test violates
\hyperref[cnd:ins-inj]{$\InsI$}, it can be split into multiple case
tests for each coincidence of instantiated variables.

\begin{example}
  \label{ex:sub-inj-split}
  Assume a case test $\ct_i$ with $\fv(\ct_i) = \Set{ x, y, z }$ that violates \nameref{cnd:ins-inj} and all free variables coincide in any combination.
  These are given by the partitions of the free variables.
  \begin{equation*}
    \begin{aligned}[c]
      &\Set[\big]{ \Set{ x, y, z } } \\
      &\Set[\big]{ \Set{ x },  \Set{ y, z } } \\
      &\Set[\big]{ \Set{ y }, \Set{ x, z } }
    \end{aligned}
    \qquad\qquad
    \begin{aligned}[c]
      &\Set[\big]{ \Set{ z }, \Set{ x, y } } \\
      &\Set[\big]{ \Set{ x }, \Set{ y }, \Set{ z } }
    \end{aligned}
  \end{equation*}

  We then need to split $\ct_i$ into five case tests in which the variables in each group are replaced by a single variable.
  For example, if $y$ and $z$ coincide, we replace each occurrence of them by a new variable $v_{y,z}$.
\end{example}

Injectivity of the instantiations also ensures that the number of instantiated variables corresponds to the number of free variables.
\begin{equation*}
  \abs[\big]{ \fv(\ct_i) \inst } = \abs[\big]{ \fv(\ct_i) }
\end{equation*}

\begin{example}
  Consider the situation of \cref{ex:db-case-tests} and a trace $t$ in which a manager $\Party{M_1}$ and the employees $\Party{E_1}$, $\Party{E_2}$ cause a violation.
  Assume there exist single-matched traces $t_1$, $t_2$ with
  \begin{align*}
    \ctr(t_1) &= \Set[\big]{ (\ct_1, \subst{ m \mapsto \Party{M_2} }) } \\
    \ctr(t_2) &= \Set[\big]{ (\ct_2, \subst{ e_i \mapsto \Party{E_3}, e_j \mapsto \Party{E_4} }) }\,,
  \end{align*}
  where only the necessary parties are corrupted.
  By \hyperref[cnd:rep-prop]{$\RepP_{\ct_1}$}, there exists a trace $t_1'$ with
  \begin{align*}
      \ctr(t_1') &= \Set[\big]{ (\ct_1, \del[ m \mapsto \Party{M_1} ]) } \\
      \cor(t_1') &= \cor(t_1)\del[ \Party{M_2} \mapsto \Party{M_1} ] = \Set{ \Party{M_1} }\,.
  \end{align*}
  By \hyperref[cnd:rep-prop]{$\RepP_{\ct_2}$}, there exists a trace $t_2'$ with
  \begin{align*}
      \ctr(t_2') &= \Set[\big]{ (\ct_2, \del[ e_i \mapsto \Party{E_1}, e_j \mapsto \Party{E_2} ] ) } \\
      \cor(t_2') &= \cor(t_2) \del[ \Party{E_3} \mapsto \Party{E_1}, \Party{E_4} \mapsto \Party{E_2} ] \\
                 &\phantom{={}} = \Set{ \Party{E_1}, \Party{E_2} }\,. \tag*{\qedhere}
  \end{align*}
\end{example}

\subsection{Soundness and Completeness}
\label{sec:sound-comp-trace-prop}

In this section, we defined the class of protocols and case tests
where we can express sufficiency as a trace property by stating
\hyperref[cnd:single]{$\SinM$}, \hyperref[cnd:ins-inj]{$\InsI$}
and
\hyperref[cnd:rep-prop]{$\RepP$}.
The latter can be checked syntactically, while the first two
can be verified directly.

Using the results we obtained above, we can now finally define the verification conditions in terms of trace properties.
In the following, we assume $\ps{P}$ to be an accountability protocol.

We write $\Suff^{\tp}$, $\VerE_{\secprop}^{\tp}$, $\VerNE_{\secprop}^{\tp}$, $\Min^{\tp}$, and $\Uniq^{\tp}$ if the respective condition holds for all $i \in [n]$.
We denote the conjunction of these conditions by $\VC_{\secprop}^{\tp}$ or by $\VC^{\tp}$ if the security property can be inferred from context.

We show the correctness of these conditions by relating them 
to the axiomatic characterization from \cref{sec:verif-acc},
which has been proven equivalent to \cref{def:apv} in \cref{thm:axiom-eq}.
\Cref{lem:ver-equiv} to \cref{lem:suff-cmpl} in
\cref{sec:sound-comp-trace-prop-proof} and
\cref{tab:ver-cond} give a nuanced picture of their relationship,
which is useful to interpret counterexamples (see also
\cref{sec:inter-results}).
\Cref{thm:snd-trace-prop} and \cref{thm:cmpl-trace-prop} in
\cref{sec:sound-comp-trace-prop-proof} show soundness and
completeness.

\section{Verifying Accountability using Tamarin}
\label{sec:automated-verif}
\label{sec:acc-lemm}
\label{sec:syntac-add}

\begin{full}
\begin{table}
  \caption{Suffixes of generated lemmas}
  \label{tab:lemma-suffixes}
  \centering
  \begin{tabular}{ll}
    \toprule
    suffix & condition \\
    \midrule
    \texttt{suff}           & \nameref{cnd:sufftp} \\
    \texttt{verif_empty}    & \nameref{cnd:vertp-e} \\
    \texttt{verif_nonempty} & \nameref{cnd:vertp-ne} \\
    \texttt{min}            & \nameref{cnd:mintp} \\
    \texttt{uniq}           & \nameref{cnd:uniqtp} \\
    \texttt{inj}            & \nameref{cnd:ins-inj} \\
    \texttt{single}         & \nameref{cnd:single} \\
    \bottomrule
  \end{tabular}
\end{table}

\end{full}

\Tamarin{} \parencite{meier2013tamarin} is a protocol verification
tool that supports falsification and unbounded verification in the
symbolic model. 
\begin{full}
Security protocols are specified using multiset-rewrite
rules, but support for specifying protocols in SAPiC
\parencite{kremer2014automated} has recently been added.
\end{full}
This makes \Tamarin{} particularly suitable for integrating our results.
We extended \Tamarin{} with two syntactical elements, case tests
and accountability lemmas.
Case tests are specified by
\begin{spreadlines}{0pt}
  \begin{flalign*}
    &\texttt{\textbf{test} ⟨name⟩:} &&\\
    &\texttt{\ \ "⟨τ⟩"} &&
  \end{flalign*}
\end{spreadlines}
where \texttt{⟨name⟩} is the name of the case test and \texttt{⟨τ⟩} is its formula.
Accountability lemmas are defined similarly to standard lemmas
\begin{spreadlines}{0pt}
  \begin{flalign*}
    &\texttt{\textbf{lemma} ⟨name⟩:} && \\
    &\texttt{\ \ ⟨name₁⟩,…,⟨nameₙ⟩ \textbf{account(s) for} "⟨φ⟩"} &&
  \end{flalign*}
\end{spreadlines}
where \texttt{⟨name⟩} is the name of the lemma, \texttt{⟨nameᵢ⟩} are the names of previously defined case tests, and \texttt{⟨φ⟩} is the security property.
The implementation allows defining an arbitrary number of
accountability lemmas. This is especially useful when experimenting
with different sets of case tests, discovering potential attacks, and
analyzing accountability properties in general.
\begin{full}
The names of the lemmas generated for an accountability lemma have the following structure
\begin{equation*}
  \texttt{⟨lemma-name⟩_⟨case-test-name⟩_⟨suffix⟩}
\end{equation*}
where \texttt{⟨suffix⟩} is named according to \cref{tab:lemma-suffixes} and \texttt{⟨case-test-name⟩} is not used for \nameref{cnd:vertp-e}.
\end{full}
Each accountability lemma consists of a set of case tests and
a security property and thus specifies a verdict function according to
\cref{def:verdict-function}.

We translate each accountability lemma into a set of standard
lemmas stating the trace properties \nameref{cnd:sufftp},
\nameref{cnd:vertp-e}, \nameref{cnd:vertp-ne}, \nameref{cnd:mintp},
\nameref{cnd:uniqtp}, \nameref{cnd:ins-inj}, \nameref{cnd:single}.
In 
\appendixorfull{sec:guardedness}
we show that all these lemmas
adhere to the guardedness requirement of \Tamarin{} provided that the
case tests are guarded.
An accountability lemma holds for a protocol $\ps{P}$ if \Tamarin{}
can successfully verify all generated lemmas and the replacement
property \hyperref[cnd:rep-prop]{$\RepP$} holds.
A protocol can be specified in terms of multiset-rewrite rules or as a SAPiC process.

When analyzing an accountability lemma,
two outcomes are possible.
Either \Tamarin{} is able to verify all conditions or at least one condition is violated.
In the latter case, it can be difficult to interpret the attack,
depending on whether the condition was necessary.
To this end, we provide a detailed decision diagram in \cref{sec:inter-results}.

\section{Case studies}\label{sec:case-studies}

\begin{table*}
    \centering
    \caption{Verification results for the DMN and MixVote case studies
        in two frameworks. We compare type of attack
        (\xmark{}=attack,$\checkmark{}$=verification),
        number of lemmas and overall verification time.
    }
    \label{tab:case-studies-new}
    \begin{tabular}{l | c@{\hspace{0.65\tabcolsep}}r@{\hspace{0.65\tabcolsep}}r | c@{\hspace{0.65\tabcolsep}}r@{\hspace{0.65\tabcolsep}}r | c@{\hspace{0.65\tabcolsep}}r@{\hspace{0.65\tabcolsep}}r | c@{\hspace{0.65\tabcolsep}}r@{\hspace{0.65\tabcolsep}}r | c@{\hspace{0.65\tabcolsep}}r@{\hspace{0.65\tabcolsep}}r }
        \toprule
        Our proposal                       & \multicolumn{3}{c}{1 role}          & \multicolumn{3}{c}{2 roles}           & \multicolumn{3}{c}{3 roles}             & \multicolumn{3}{c}{4 roles}              & \multicolumn{3}{c}{5 roles}              \\
        \midrule
        Basic DMN (duplicate ciphertexts)  & \multicolumn{3}{c|}{---}            & \multicolumn{3}{c|}{---}              & \checkmark{} & 13 & \SI{26}{\second}    & \multicolumn{3}{c|}{---}                 & \multicolumn{3}{c}{---}                  \\
        DMN + message tracing (first)      & \checkmark{} &  7 & \SI{8}{\second} & \checkmark{} &  7 & \SI{124}{\second} & \checkmark{} &  7 & \SI{1373}{\second}  & \checkmark{} &   7 & \SI{14178}{\second} & \checkmark{} & 7 & \SI{134160}{\second} \\
        DMN + message tracing (all)        & \checkmark{} &  7 & \SI{6}{\second} & \xmark{}     &  7 & \SI{12}{\second}  & \xmark{}     &  7 & \SI{22}{\second}    & \xmark{}     &   7 & \SI{100}{\second}   & \xmark       & 7 & \SI{355}{\second}     \\
        MixVote (unbounded)                & \checkmark{} & 14 & \SI{6}{\second} & \multicolumn{3}{c|}{---}              & \multicolumn{3}{c|}{---}                & \multicolumn{3}{c|}{---}                 & \multicolumn{3}{c}{---}                  \\
        \midrule[\heavyrulewidth]
        \cite{kuennemann2019automated} & \multicolumn{3}{c}{1 party}         & \multicolumn{3}{c}{2 parties}         & \multicolumn{3}{c}{3 parties}           & \multicolumn{3}{c}{4 parties}            & \multicolumn{3}{c}{5 parties}                                                    \\
        \midrule
        DMN + message tracing (first)      & \checkmark{} &  7 & \SI{7}{\second} & \checkmark{} & 17 & \SI{133}{\second} & \checkmark{}  & 46 & \SI{2146}{\second} & \checkmark{}                & 149 & \SI{23827}{\second}      & \multicolumn{1}{c}{---$^*$} &  544 & \multicolumn{1}{c}{---} \\
        DMN + message tracing (all)        & \checkmark{} &  7 & \SI{4}{\second} & \xmark{}     & 17 & \SI{23}{\second}  & \xmark{}      & 46 & \SI{115}{\second}  & \xmark{}                    & 149 & \SI{548}{\second}        & \xmark                      &  544 & \SI{2922}{\second}      \\
        MixVote (unbounded)$^{**}$         & \checkmark{} & 14 & \SI{5}{\second} & \checkmark{} & 34 & \SI{58}{\second}  & \checkmark{}  & 92 & \SI{2721}{\second} & \multicolumn{1}{c}{---$^*$} & 298 & \multicolumn{1}{c|}{---} & \multicolumn{1}{c}{---$^*$} & 1112 & \multicolumn{1}{c}{---} \\
        \midrule[\heavyrulewidth]
        \multicolumn{16}{l}{
        \footnotesize
        $^*$ No verification results due to memory exhaustion.
        $^{**}$ Each party acts in the same role, that of the server.
    } \\
        \bottomrule
    \end{tabular}%
\end{table*}
\begin{table*}
    \centering
    \caption{Verification results for case studies from \cite{kuennemann2019automated} in the unbounded setting.}
    \label{tab:case-studies-previous}
    \begin{tabular}{l | c@{\hspace{0.65\tabcolsep}}r@{\hspace{0.65\tabcolsep}}r | l@{\hspace{0.65\tabcolsep}}r@{\hspace{0.65\tabcolsep}}r }
        \toprule
                                                                                                 & \multicolumn{3}{c}{Our proposal}   & \multicolumn{3}{c}{\cite{kuennemann2019automated}} \\    
        \midrule
        WhoDunit (fixed)                    & \checkmark{}                  & 7                  & \SI{52}{\second}                   & \checkmark{} ($\rel_c)$ &  8 & \SI{24}{\second}  \\
                                            &                               &                    &                                    & \checkmark{} ($\rel_w)$ &  7 & \SI{11}{\second}  \\
        Certificate Transparency (extended) & \checkmark{}                  & 27                 & \SI{17}{\second}                   & \checkmark{}            & 31 & \SI{21}{\second}  \\
        OCSP Stapling (trusted resp.)       & \checkmark{}                  & 7                  & \SI{1}{\second}                    & \checkmark{}            &  7 & \SI{515}{\second} \\
        OCSP Stapling (untrusted resp.)     & \xmark{}                      & 7                  & \SI{1}{\second}                    & \xmark{}                &  7 & \SI{75}{\second}  \\
        \bottomrule
    \end{tabular}%
\end{table*}

We demonstrate our methodology on eight case studies, four from prior work~\cite{kuennemann2019automated} and four more in the domain of mixnets and electronic voting.
We summarize our findings in \cref{tab:case-studies-previous,tab:case-studies-new}.
For each case study, we provide the verification result (\checkmark{} for successful verification, \xmark{} if we found an attack), the number of generated lemmas, and the time needed to verify all lemmas (even if an attack is found).

Before describing the case studies, we want to emphasize the importance of distinguishing between sessions, roles, and parties.
The number of sessions specifies how many instances of a protocol can be executed in parallel.
In each session, there can be multiple roles---for example, a server or a client---with different frequencies.
Within a protocol trace, these roles are instantiated with concrete party identifiers
drawn from a countably infinite set of public names.
Depending on the protocol, 
a party may participate in multiple sessions 
and
each session may be run by different sets of parties.
Even if a protocol has just one role, an unbounded number of parties
may be involved.

\subsection{Case studies from~\texorpdfstring{\cite{kuennemann2019automated}}{Künnemann, Esiyok, and Backes (2019)}}

We briefly recall the case studies from prior work~\cite{kuennemann2019automated}.
\emph{WhoDunit} illustrates a situation where a third party $J$ cannot provide a correct verdict.
$S$ sends some value to $A$ and $J$ and $A$ should forward it to $J$.
We are interested in accountability for $J$ receiving the same value from $A$ and $S$.
Without signatures it is impossible to distinguish between $A$ tampering with the message that it should forward and $S$ sending different values to $A$ and $J$.
The fixed version uses signatures to give evidence of provenance.
We extended the fixed version to an unbounded number of parties in the roles of $A$ and $S$.
The original version considers only a single communication session, hence both the analysis with respect to $\rel_c$ and $\rel_w$ (see \cref{sec:count-rel}) run faster, because they need to consider only a very small number of possible interleavings (three protocol messages).

\emph{Certificate Transparency}~\cite{rfc6962} is an accountability protocol that provides transparency for a public key infrastructure.
\textcite{kuennemann2019automated} extended a simple model~\cite{bruni2017automated} for a single certificate authority and a single logging authority.
We adapted the model to allow for an unbounded number of both, but otherwise adhered to their original formulation.
We observe a slight speed up in the verification, which is likely due to the removal of logical redundancies in the axiomatic characterization shown in \cref{sec:verif-acc}.
In contrast to WhoDunit, the original model already considered an unbounded number of interactions between concrete parties.
Hence the proofs are similarly structured.

\emph{OCSP Stapling}~\cite{rfc6066} is a mechanism to attach signed Online Certificate Status Protocol (OCSP~\cite{rfc6960}) messages during a TLS handshake.
The server's goal is to provide evidence that their certificate has not been revoked recently without the client exposing their browsing behavior to the OCSP server.
The model from \textcite{kuennemann2019automated} used an explicit clock process to model time.
We extended their model to an unbounded number of clients, TLS servers, and OCSP servers.
Moreover, we ported their SAPiC model to multiset-rewrite rules to exploit a more effective modeling of timepoints, improving the verification time by two orders of magnitude.
Otherwise, in particular concerning the communication, we remained faithful to their modeling.
The new model of timepoints avoids the use of helping lemmas compared to three required previously.
It also reduces the verification time by two orders of magnitude in the case the OCSP responder is trusted and accountability holds and by at least one in case the OCSP responder is untrusted.

\subsection{Mixnets}

Mixnets are a building block for many privacy-preserving technologies, e.g., e-voting systems, anonymous messaging, anonymous routing, and oblivious RAM (see a recent survey~\cite{hainesSoKTechniquesVerifiable2020a}).
While basic mixnets are only suitable in the honest-but-curious attacker model, they can be extended to provide verifiability and even accountability.

In this work, we focus on basic decryption mix nets (basic DMN) and their extension with message tracing (DMN + message tracing) as proposed in~\cite{hainesSoKTechniquesVerifiable2020a}.
To the best of our knowledge, this case study provides the first automated formal verification results for these kinds of DMNs.

In a DMN, each sender encapsulates their plaintext within several layers of encryption using the last mix server's public key first and the first mix server's public key last.
Each mix server decrypts the messages it receives, removing the outermost layer.
Each mix server shuffles the messages before sending them to the next server on a public channel.
For accountability, we assume these messages to be stored on a public append-only bulletin board that cannot be tampered with.

In basic DMNs, the ciphertexts on the bulletin board are continuously checked for duplicates and, in the case of a duplicate, the protocol is terminated.
Depending on which phase they were posted in, this audit correctly identifies the responsible mix server or sender.

In DMNs with message tracing, the senders store the random coins they used for encryption and each intermediate ciphertext they produce.
During the audit, each sender verifies for each mixing step that their intermediate ciphertexts appears on the bulletin board.
If this is not the case, the sender in question uses their stored random coins to prove that the mix server misbehaved.
We consider two cases of
DMNs with message tracing: in the first, the sender stops the audit once the first misbehaving mix server is found.  In the second variant, the audit continues until the last mix server.%

For this case study, we modeled basic DMNs and two variants of DMNs with message tracing in \Tamarin{}.
For the former, we allow three mix server roles and two sender roles.
For the latter, we fix the number of sender roles to two, and scale the number of mix server roles, the \emph{mix length}, from one to five.
Note that there is still an unbounded number of sessions with an
unbounded number of potential senders and mix servers in these roles,
similar to how the Tor network fixes the number of onion routers to
three, but has millions of users.%

For the basic DMN, we can show accountability for duplicate ciphertexts when we limit the senders and mix servers to only duplicate messages, but not submit otherwise dishonestly generated messages.
We define two case tests, one for senders, which checks for duplicates in the senders' output, and one for mix servers, which checks for duplicates in the mix servers' output including the final output.
The tests hold any party accountable that posts a ciphertext that has already been posted on the bulletin board.
Together, they provide accountability for the property that no duplicates occur in the same phase of a session.

For DMN with message tracing, we define a single case test holding the mix servers accountable that have been identified by a sender during the audit
\begin{equation}
    \label{eq:case-study-dmn-ct}
    \ct \defas \Exists{ \var{sid}, \var{x}, \var{i} } m = \del<\var{sid}, \var{x}> \land \fact{B}(m)@\var{i}\,,
\end{equation}
where $\fact{B}(\del<\var{sid}, \var{m}>)$ denotes that a server blamed the mix server $\var{m}$ in session $\var{sid}$.
We note that the variable $\var{m}$ is free in the case test. Hence, the parties in the verdict are pairs consisting of
a session identifier and a mix server. 
This allows a mix server to be honest in one session and dishonest in another.%

Up to a mix length of five,
we can show accountability when the senders/auditors
only blame the first mix server they catch cheating.
This confirms an existing formal result in the 
cryptographic model~\cite{kustersSElectLightweightVerifiable2016}.
For the variant where they blame all mix servers who have not posted
the correct intermediate ciphertext on the bulletin board, we find
a counterexample up to a mix length of five---with one exception. If
there is only one mix server role, this case is equivalent to the
other variant and accountability holds.
For a mix length of two or more, we find that uniqueness is
violated, indicating that a mix server can be blamed despite acting
honestly.
This happens when a dishonest mix server tampers with the ciphertext
in one of the previous stages, as the mix servers down the line will
themselves produce ciphertexts that fail the audit.%

For comparison and to
evaluate the impact of using case tests instead of defining the
verdict function explicitly,
we ported the two variants of DMNs with message tracing to the
framework of \cite{kuennemann2019automated}.
First, we had to limit the number of sessions to one. 
Listing all pairs of mix server identities (which include session
identifiers) would have been impossible.

Comparing the results of our approach with the results of \cite{kuennemann2019automated} in 
\cref{tab:case-studies-new}, we see that they agree on the outcome.
We note, however, that while the number of generated lemmas stays constant with an increasing number of mix servers in our approach, they increase exponentially in the other.
This is, again, due to the explicit enumeration of all cases in the
verdict function in \cite{kuennemann2019automated}. Even though we
fix the identities to the number of roles, i.e., 
$M1$ to $M3$ in the case of a mix length of three, 
we have to account for each combination of mix servers in the verdict function, e.g.,
\begin{align*}
\vf(t) \defas \begin{cases}
    \verdict{(M1)}             & \text{ if } \omega_{M1}(t) \\
    \verdict{(M2)}             & \text{ if } \omega_{M2}(t) \\
    \verdict{(M3)}             & \text{ if } \omega_{M3}(t) \\
    \verdict{(M1), (M2)}       & \text{ if } \omega_{M1,M2}(t) \\
    \verdict{(M1), (M3)}       & \text{ if } \omega_{M1,M3}(t) \\
    \verdict{(M2), (M3)}       & \text{ if } \omega_{M2,M3}(t) \\
    \verdict{(M1), (M2), (M3)} & \text{ if } \omega_{M1,M2,M3}(t) \\
    \mkern2mu\verdict{}        & \text{ otherwise}\,,
\end{cases}
\end{align*}
Here, each $\omega_S$ is a trace property that is satisfied if and only
if the annotated mix servers in $S$ are blamed.
The number of cases in the verdict function equals the cardinality of
the powerset of the set of \enquote{blameable} parties, which grows
exponentially.
Hence, scalability is severely limited by this approach.%

Thanks to the use of case tests, our approach permits the
specification of the verdict function independent of the number of
parties and even the number of mix server roles,
keeping the user's specification effort minimal.
Another consequence is that the accountability lemmas we produce are
actually the same. 
We nevertheless observe an increase in verification time with the mix
length, but this is expected, as the backward-resolution approach in
\Tamarin{} has to explore a larger state space. While a smarter
encoding of the case study might be possible---we tried several---this
effect would likely occur when verifying other properties, e.g.,
correspondence of output and input, in the same model.
Compared to the previous approach, we see that the verification time is
drastically reduced, sometimes by a factor of five. This is
despite the restriction to a fixed set of parties and a single
session for the previous approach. The difference is more pronounced
the longer the mixnet is, which can be explained as follows: When the
mix lengths is increased, the search space for the backward-resolution
is increased, which, in both approaches, affects the verification time
\emph{per lemma}. Since a lot more lemmas need to be proven for the
previous approach, the effect is amplified by a factor that increases
with the mix length.%

\subsection{Dispute resolution in MixVote}

MixVote~\cite{DBLP:conf/csfw/BasinRS20} adds a dispute resolution procedure to the mixnet-based voting protocol Alethea~\cite{basin2018alethea}.

We first give a high-level overview of the protocol.
A voter $H$ uses their device $D$ to compute their ballot by first encrypting their vote under the voting authority's (server $S$) public key which is then signed by $D$.
The voter casts their ballot $b$ by submitting it to some platform $P$, which forwards it over the network to $S$.
The ballot is verified by $S$ by checking whether it contains a signature corresponding to an eligible voter who has not previously voted.
In this case, $b$ is added to the list of recorded ballots $[b]$.
$S$ then signs $b$ and sends it back to $H$ as an evidence that $b$ was indeed received by the authority.
This evidence is kept by $H$ in the case for later disputes.
Once the voting phase is over, $S$ computes the tally of the recorded ballots $[b]$ by decrypting them using a mixnet.
Finally, $S$ publishes the encrypted ballots and the decrypted votes on the public bulleting board such that a voter can verify that their ballot is included.

This case study shows that our approach can be applied to existing specifications with minimal effort and is based on one of the \Tamarin{} models from \textcite{DBLP:conf/csfw/BasinRS20} ({\mbox{mixvote_SmHh}).
In this model, the voting authority $S$ can be corrupted while the voters are honest.
When corrupted, the authority's secret key is given to the adversary and the incoming and outgoing channels are modeled as insecure.
$S$ is partially trusted to sign and return a valid ballot received from $P$.
This model covers the case where a voter $H$ claims to have cast a ballot $b$ while the authority $S$ claims that this is not the case.

The original model runs a single session of the protocol with the identity of the server fixed to `$\cnt{S}$'.
We extended the model to support an unbounded number of sessions, used an unreliable insecure channel from $P$ to $S$, and added a corruption mechanism for the server.
Note that the restriction to a single server role per session is a property of the protocol and not a limitation of our approach.%

We focus on the two properties protecting an honest voter in the case of a dispute:
\begin{itemize}
  \item $\VoterC$: ensures that whenever an honest voter detects that one of their ballots was not recorded correctly, they can convince others that $S$ is dishonest.
  \item $\TimelyP$ ensures that whenever an honest voter casts a ballot, they cannot be prevented from continuing the protocol until their ballot is recorded or they can convince others that $S$ is dishonest.
\end{itemize}
We define an accountability lemma for each property.

\paragraph{Accountability for $\VoterC$}
We first define the security property, which is directly encoded in $\VoterC$.
Whenever an honest voter $H$ validates their ballot, it is indeed included in the list of recorded ballots on the bulletin board.
\begin{align*}
  \secprop_{\VoterC} {}\defas{} &\Forall{ \var{H}, \var{b}, \var{b1}, \var{i} } \fact{Verify}(\var{H}, \var{b}, \var{b1})@i \\
                                &\impliestp \Exists{ \var{BB}, \var{j} } \fact{BB}_{\mathit{rec}}(\var{BB}, \del< \mlq\cnt{b}\mrq, \var{b} >)@j
\end{align*}
We then define a case test which blames $S$ whenever a ballot is recorded on the bulletin board that has not been signed by the voter's device $D$ and the verifiability check is reached.

\paragraph{Accountability for $\TimelyP$}
The security property follows with a slight change from $\TimelyP$.
Whenever an honest voter $H$ casts a ballot and all relevant information is published on the bulletin board, the ballot is indeed included in the list of recorded ballots on the bulletin board.
\begin{align*}
  \secprop_{\TimelyP} {}\defas{} &\Forall{ \var{H}, \var{b}, \var{i}, \var{j} } \fact{Ballot}(\var{H}, \var{b})@i \land \fact{End}@j \\
                                 &\begin{aligned}[t] 
                                    \impliestp \Exists{ \var{BB}, \var{k} } &\fact{BB}_{\mathit{rec}}(\var{BB}, \del< \mlq\cnt{b}\mrq, \var{b} >)@k \\
                                   {}\land{} &i < k \land k = j
                                 \end{aligned}
\end{align*}
We define a case test which blames $S$ whenever a ballot is recorded on the bulletin board which has not been signed by the voter's device $D$ and the point where the voter receives their ballot from $D$ is reached.
We note that this case test is the same as the one for $\VoterC$ with the exception of the point in the protocol needed to be reached.
Here, we have to slightly strengthen the accountability lemma compared to $\TimelyP$.
We move the requirement that the ballot is cast before the voting ends ($i < k$) from the premise to the conclusion.
Otherwise, when a ballot is cast after the vote has ended, we have a matching case test without a security violation, i.e., a counterexample to \hyperref[cnd:vertp-ne]{$\VerNE^{\tp}$}.

Our approach can automatically show that accountability holds for the two properties described above without requiring helping lemmas.

We ported the model to the framework of \cite{kuennemann2019automated} to provide a comparison with our approach.
This version also supports an unbounded number of sessions, but due to the restriction on concrete party identifiers, we had to limit the set of parties that could act as the server.
We analyzed the protocol with up to five distinct server parties and obtained results with up to three.
In the case of four and five identities, \Tamarin{}'s search algorithm exceeded the amount of available memory.%

The results in the framework of \cite{kuennemann2019automated} agree with the results of our approach,
but the time needed to obtain them increases exponentially with the
number of parties, whereas the result presented here holds for an
unbounded number of parties.%

For future work, it might be interesting to hold both the voter and server accountable at the same time,
by merging the
MixVote model that covers a dishonest voter and an honest authority (\mbox{mixvote_ShHm}) 
with the one we investigated here.

\section{Conclusion}
\label{sec:conclusion}

In this work, we provide an automated verification methodology for
accountability that supports an unbounded number of participants, and
thus an unbounded number of security violations. 
This precludes explicit assignment of blame. We therefore 
introduced case tests---a higher-level variant of 
\citeauthor{bruni2017automated}'s accountability tests---
and used them to define highly flexible verdict functions.
Our approach also improves readability, as we may consider each case test as a specific
manifestation of a violation.
We showed how the verdict-based verification conditions can be
expressed using case tests and finally be encoded in terms of trace
properties.
Furthermore, we extended \Tamarin{} with the ability to automatically generate these from accountability lemmas.
Our case studies demonstrate applications for transparency protocols,
revocation protocols, mixnets, and dispute resolution in e-voting.

\paragraph*{Acknowledgements}
This research was partly supported by the ERC Synergy Grant ``imPACT'' (No. 610150).

\printbibliography

\appendix

\subsection{Proofs for \texorpdfstring{\cref{sec:verif-acc}}{Section IV}}
\label{sec:verif-acc-proof}

\begin{conf}
  For the proofs see \appendixorfull{sec:verif-acc-proof}.
\end{conf}

\subsubsection{Helping lemmas}

The following lemma will help us in the soundness proof.
Assume actual and counterfactual traces $t$, $t'$ are related and $S'$ is a set in $\apv(t'$).
Then $\apv(t)$ is empty or there exists a subset of $S$ which is in $\apv(t)$.
\begin{lemma}
  \label{lem:rel-apv}
  For all traces $t$, $t'$
  \begin{align*}
    \rel(t,t') \land S' \in \apv(t') \implies{} &\apv(t) = \verdict{} \\
                                       {}\lor{} &\Exists{ S } S \in \apv(t) \land S \subseteq S'\,.
  \end{align*}
\end{lemma}
\begin{full}
\begin{proofcontra}
  Assume $\rel(t,t')$, $S' \in \apv(t')$, $\apv(t) \neq \verdict{}$ and there does not exist $S \in \apv(t)$ such that $S \subseteq S'$.
  From the latter follows $S' \notin \apv(t)$.
  For this to be the case, at least one of the three requirements in \cref{def:apv} has to be violated.
  From $\apv(t) \neq \verdict{}$ follows with \cref{cor:apv} that $t \holds \neg \secprop$.
  Thus \cref{eq:apv-req-1} is satisfied.
  As $S' \in \apv(t')$, there exists a trace $t\sy{3}$ such that
  \begin{equation*}
    t\sy{3} \holds \neg \secprop \land \cor(t\sy{3}) = S' \land \rel(t',t\sy{3})\,.
  \end{equation*}
  With $\rel(t,t')$ follows $\rel(t,t\sy{3})$ and thus \cref{eq:apv-req-2} is fulfilled.
  Therefore, \cref{eq:apv-req-3} cannot hold.
  Let $t''$ be such that $S = \cor(t'')$ is minimal.
  Then $S \in \apv(t)$ and $S \subsetneq S'$.
  However, this contradicts our assumption that no such set exists.
\end{proofcontra}
\end{full}
It may seem unintuitive that either the apv is empty or for each set of parties in the apv of the counterfactual trace a subset of this set must exist in the apv of the actual trace.
We note that the minimality requirement of the apv is weaker in the counterfactual trace than in the actual trace.
The traces related to $t'$ are a subset of the traces related to $t$ and thus there may be a trace related to $t$ showing that a set $S$ is not minimal, but this trace is not necessarily related to $t'$.

From \cref{eq:apv-req-3} of the apv, we derive that the apv does not contain two sets where one is a strict subset of the other.
\begin{corollary}
  \label{cor:min-apv}
  For all traces $t$ and sets $S$
  \begin{equation}
    \label{eq:cor-apv-min}
    S \in \apv(t) \implies \NExists{ \qs{S'} } S' \in \apv(t) \land S' \subsetneq S\,.
  \end{equation}
\end{corollary}
\begin{full}
\begin{proofcontra}
  Assume \cref{eq:cor-apv-min} does not hold.
  Then there exist $S$, $S' \in \apv(t)$ with $S' \subsetneq S$.
  From $S' \in \apv(t)$ follows the existence of a trace $t'$ such that
  \begin{equation*}
    t' \holds \neg \secprop \land \cor(t') = S' \land \rel(t,t')\,.
  \end{equation*}
  As $S' \subsetneq S$ this violates \cref{eq:apv-req-3} with respect to $S$.
\end{proofcontra}
\end{full}

\subsubsection{Soundness and Completeness}
\label{sec:sound-compl}

We show that the verdict-based verification conditions are sound and complete with respect to \cref{def:acc}.

\begin{theorem}[Soundness]
  \label{thm:snd-set}
  For any protocol $\ps{P}$, security property $\secprop$, and verdict function $\vfup$, if $\VC$ holds, then $\vfup$ provides $\ps{P}$ with accountability for $\secprop$.
\end{theorem}
\begin{full}
\begin{proof}
  Assume $\VC$ holds.
  We show that for all traces $t$, $\apv(t) = \vf(t)$.
  Let $t$ be an arbitrary trace.

  From \cref{cor:apv} and \nameref{cnd:ver} directly follows $\apv(t) = \verdict{} \iff \vf(t) = \verdict{}$.
  Hence, we only have to consider nonempty verdicts in the following.
  The proof consists of two parts.
  We first show that $\apv(t) \subseteq \vf(t)$ and then $\vf(t) \subseteq \apv(t)$.

  Assume $S \in \apv(t)$.
  To show that $S \in \vf(t)$, we have to prove that $S$ satisfies \cref{eq:conj-comp-1,eq:conj-comp-2,eq:conj-comp-3,eq:conj-comp-4}.
  \begin{conditions}
    \item[\Cref{eq:conj-comp-1}:]
      From \cref{eq:apv-req-2} follows the existence of a trace $t'$ such that
      \begin{equation*}
        t' \holds \neg \secprop \land \cor(t') = S \land \rel(t,t')\,.
      \end{equation*}
      It suffices to show that $\vf(t') = \verdict{ S }$.
      From \nameref{cnd:ver} and $t' \holds \neg \secprop$ follows $\vf(t') \neq \verdict{}$.
      Assume $\abs{ \vf(t') } \geqslant 2$.
      Then there exist $S\sy{3}$, $S\sy{4}$ such that $\verdict{ S\sy{3}, S\sy{4} } \subseteq \vf(t')$ and $S\sy{3} \neq S\sy{4}$.
      By \hyperref[cnd:suff]{$\Suff_{S\sy{3}}$} and \hyperref[cnd:suff]{$\Suff_{S\sy{4}}$}, there exist traces $t\sy{3}$, $t\sy{4}$ such that
      \begin{alignat*}{3}
        \vf(t\sy{3}) = \verdict{ S\sy{3} } &\land \cor(t\sy{3}) \subseteq S\sy{3} &&\land \rel(t',t\sy{3}) \\
        \vf(t\sy{4}) = \verdict{ S\sy{4} } &\land \cor(t\sy{4}) \subseteq S\sy{4} &&\land \rel(t',t\sy{4})\,.
      \end{alignat*}
      From \hyperref[cnd:uniq]{$\Uniq_{S\sy{3}}$} and $\rel(t',t\sy{3})$ follows $S\sy{3} \subseteq S$.
      From \hyperref[cnd:uniq]{$\Uniq_{S\sy{4}}$} and $\rel(t',t\sy{4})$ follows $S\sy{4} \subseteq S$.
      Since $S\sy{3} \neq S\sy{4}$, either $S\sy{3} \subsetneq S$ or $S\sy{4} \subsetneq S$.
      However, as $S \in \apv(t)$, $\rel(t,t\sy{3})$, and $\rel(t,t\sy{4})$, this would violate the minimality of $S$.
      Thus $S\sy{3} = S\sy{4} = S$ and $\vf(t') = \verdict{ S }$.
    \item[\Cref{eq:conj-comp-2}:]
      Assume \cref{eq:conj-comp-2} does not hold.
      Then there exists $S' \in \vf(t)$ such that $S' \subsetneq S$.
      We argue that $S' \in \apv(t)$ by showing that all three requirements of \cref{def:apv} are satisfied.
      By \hyperref[cnd:suff]{$\Suff_{S'}$}, there exists a trace $t'$ such that
      \begin{equation*}
        \vf(t') = \verdict{ S' } \land \cor(t') \subseteq S' \land \rel(t,t')\,.
      \end{equation*}
      From \hyperref[cnd:uniq]{$\Uniq_{S'}$} follows $S' \subseteq \cor(t')$ and thus with the results from above $\cor(t') = S'$.
      From \nameref{cnd:ver} follows $t' \holds \neg \secprop$.
      Thus $t'$ satisfies \cref{eq:apv-req-2}.
      Since $S \in \apv(t)$, \cref{eq:apv-req-1} is also satisfied.
      If \cref{eq:apv-req-3} would not be fulfilled, then there would exist a trace $t''$ such that
      \begin{equation*}
        t'' \holds \neg \secprop \land \cor(t'') \subsetneq S' \land \rel(t,t'')\,.
      \end{equation*}
      As $S' \subsetneq S$, this would violate the minimality of $S$.
      Thus \cref{eq:apv-req-3} holds and $S' \in \apv(t)$.
      However, this violates \cref{cor:min-apv}.
    \item[\Cref{eq:conj-comp-3}:]
      By \cref{eq:apv-req-2}, there exists a trace $t'$ such that $\cor(t') = S$ and $\rel(t,t')$.
      From \cref{eq:rel-corrupted} follows $S \subseteq \cor(t)$.
    \item[\Cref{eq:conj-comp-4}:]
      \Cref{eq:conj-comp-4} follows directly from \cref{eq:apv-req-1}.
  \end{conditions}

  We now consider the reverse direction.
  Assume $S \in \vf(t)$.
  To show that $S \in \apv(t)$, we have to prove that $S$ satisfies \cref{eq:apv-req-1,eq:apv-req-2,eq:apv-req-3} of the apv.
  \begin{conditions}
    \item[\Cref{eq:apv-req-1}:]
      From \nameref{cnd:ver} directly follows $t \holds \neg \secprop$.
    \item[\Cref{eq:apv-req-2}:]
      By \nameref{cnd:suff}, there exists a trace $t'$ such that
      \begin{equation*}
        \vf(t') = \verdict{ S } \land \cor(t') \subseteq S \land \rel(t,t')\,.
      \end{equation*}
      From \nameref{cnd:uniq} follows $S \subseteq \cor(t')$ and thus $\cor(t') = S$.
      From \nameref{cnd:ver} follows $t' \holds \neg \secprop$.
      Hence, $t'$ satisfies \cref{eq:apv-req-2}.
    \item[\Cref{eq:apv-req-3}:] 
      Assume the \lcnamecref{eq:apv-req-3} does not hold.
      Then there exists a trace $t'$ such that
      \begin{equation*}
        t' \holds \neg \secprop \land \cor(t') \subsetneq S \land \rel(t,t')\,.
      \end{equation*}
      Let $t'$ be minimal with respect to $S' = \cor(t')$.
      Then $S' \in \apv(t')$ and $S' \subsetneq S$.
      Since $\rel(t,t')$, by \cref{lem:rel-apv}, $\apv(t) = \verdict{}$ or there exists $S'' \in \apv(t)$ such that $S'' \subseteq S'$.
      If $\apv(t) = \verdict{}$, it follows from the former proof that $\vf(t) = \verdict{}$, which contradicts our assumption that $S \in \vf(t)$.
      In the other case, the former proof implies $S'' \in \vf(t)$.
      However, as $S'' \subsetneq S$ this violates \nameref{cnd:min}. \qedhere
  \end{conditions}
\end{proof}
\end{full}

\begin{theorem}[Completeness]
  \label{thm:cmpl-set}
  For any protocol $\ps{P}$, security property $\secprop$, and verdict function $\vfup$, if $\vfup$ provides $\ps{P}$ with accountability for $\secprop$, then $\VC$ holds.
\end{theorem}
\begin{full}
\begin{proof}
  Assume that for all traces $t$, $\apv(t) = \vf(t)$.
  We have to show that $\VC$ holds.
  Let $t$ be an arbitrary trace.
  \nameref{cnd:ver} follows from $\apv(t) = \vf(t)$ and \cref{cor:apv}.
  Hence, we only have to consider nonempty verdicts in the following.

  \begin{conditions}
    \item[\nameref{cnd:suff}:]
      Assume $S \in apv(t)$ and $S \in \vf(t)$.
      From \cref{eq:apv-req-2} follows the existence of a trace $t'$ such that
      \begin{equation*}
        t' \holds \neg \secprop \land \cor(t') = S \land \rel(t,t')\,.
      \end{equation*}
      It suffices to show that $\vf(t') = \verdict{ S }$.
      If $\vf(t') = \verdict{}$ and thus $\apv(t') = \verdict{}$, then $t' \holds \neg \secprop$ which would violate \cref{cor:apv}.
      Assume $\abs{ \vf(t') } \geqslant 2$.
      Then there exist $S\sy{3}$, $S\sy{4}$ such that $\verdict{ S\sy{3}, S\sy{4} } \subseteq \apv(t')$ and $S\sy{3} \neq S\sy{4}$.
      By \cref{eq:apv-req-2}, there exist traces $t\sy{3}$, $t\sy{4}$ such that
      \begin{alignat*}{3}
        t\sy{3} \holds \neg \secprop &\land \cor(t\sy{3}) = S\sy{3} &&\land \rel(t',t\sy{3}) \\
        t\sy{4} \holds \neg \secprop &\land \cor(t\sy{4}) = S\sy{4} &&\land \rel(t',t\sy{4})\,.
      \end{alignat*}
      From $\rel(t',t\sy{3})$ follows $S\sy{3} \subseteq S$ and from $\rel(t',t\sy{4})$ follows $S\sy{4} \subseteq S$.
      Since $S\sy{3} \neq S\sy{4}$, either $S\sy{3} \subsetneq S$ or $S\sy{4} \subsetneq S$.
      However, as $S \in \apv(t)$ and $\rel(t,t\sy{3})$, $\rel(t,t\sy{4})$, this would violate the minimality of $S$.
      Thus $S\sy{3} = S\sy{4} = S$ and $\vf(t') = \verdict{ S }$.
      Hence, \nameref{cnd:suff} holds.
    \item[\nameref{cnd:min}:]
      Assume $S \in apv(t)$ and $S \in \vf(t)$.
      \nameref{cnd:min} follows from $\apv(t) = \vf(t)$ and \cref{cor:min-apv}.
    \item[\nameref{cnd:uniq}:]
      Assume $S \in apv(t)$ and $S \in \vf(t)$.
      By \cref{eq:apv-req-2}, there exists a trace $t'$ such that $\cor(t') = S$ and $\rel(t,t')$.
      From \cref{eq:rel-corrupted} follows $S \subseteq \cor(t)$.
    \item[\nameref{cnd:comp}:]
      Let $S$ be such that \cref{eq:conj-comp-1,eq:conj-comp-2,eq:conj-comp-3,eq:conj-comp-4} are satisfied.
      Assume $S \notin \vf(t)$ and thus by assumption $S \notin \apv(t)$.
      We provoke a contradiction by showing that $S$ satisfies \cref{eq:apv-req-1,eq:apv-req-2,eq:apv-req-3}.

      From \cref{eq:conj-comp-4} directly follows \cref{eq:apv-req-1}.

      From \cref{eq:conj-comp-1} follows the existence of a trace $t'$ such that
      \begin{equation*}
        \vf(t') = \verdict{ S } \land \cor(t') \subseteq S \land \rel(t,t')\,.
      \end{equation*}
      By assumption $\apv(t') = \verdict{ S }$ and by \cref{eq:apv-req-1} $t' \holds \neg \secprop$.
      Along with \cref{eq:conj-comp-3} follows $\cor(t') = S$ and thus $t'$ satisfies \cref{eq:apv-req-2}.

      Assume \cref{eq:apv-req-3} does not hold.
      Then there exists a trace $t''$ such that
      \begin{equation*}
        r(t,t'') \land \cor(t'') \subsetneq S \land t'' \holds \neg \secprop\,.
      \end{equation*}
      Let w.l.o.g.\ $t''$ be a trace such that $\cor(t'') = S'$ is minimal.
      Then $S' \in \apv(t)$ and $S' \subsetneq S$.
      However, by assumption it follows that $S' \in \vf(t)$ violating \cref{eq:conj-comp-2}.

      Hence, $S \in \vf(t)$. \qedhere
  \end{conditions}
\end{proof}
\end{full}
 
\subsection{Soundness and completeness of verification
conditions}\label{sec:sound-comp-trace-prop-proof}

\begin{conf}
  For the proofs see \appendixorfull{sec:sound-comp-trace-prop-proof}.
\end{conf}
\begin{full}
In the transformations we perform below, we use
\cref{eq:ct-ex-quan,eq:ct-all-quan,eq:strict-sub-tp,eq:cor-sup-tp,eq:cor-sub-tp,def:satis-rel,cor:verdict-function}.
\end{full}

\begin{lemma}[Verifiability]
  \label{lem:ver-equiv}
  \begin{equation*}
    \hyperref[cnd:ver]{\Ver_{\secprop}} \iff \hyperref[cnd:vertp-e]{\VerE_{\secprop}^{\tp}} \land \hyperref[cnd:vertp-ne]{\VerNE_{\secprop}^{\tp}}
  \end{equation*}
\end{lemma}
\begin{full}
\begin{proof}
  \begingroup
  \allowdisplaybreaks[1]
  \begin{align*}
    \Ver_{\secprop} &\ftequiv \vf(t) = \verdict{} \iff t \holds \secprop \\
                    &\ftequiv \del[\big]( \NExists{ i, \inst} t \holds \ct_i \inst) \iff t \holds \secprop \\
                    &\ftequiv \begin{aligned}[t]
                                &\del[\Big][ \del[\Big]( \oskip\And[r]{ i \in [n] } \NExists{ \inst } t \holds \ct_i \inst ) \implies t \holds \secprop ] \\
                      {}\land{} &\del[\Big][ \del[\Big]( \oskip\Or[r]{ i \in [n] } \Exists{ \inst } t \holds \ct_i \inst ) \implies t \holds \neg \secprop ]
                    \end{aligned} \\
                    &\ftequiv \begin{aligned}[t]
                                &\del[\Big][ \del[\Big]( \oskip\And[r]{ i \in [n] } t \holds \NExists{ \vec{v} } \ct_i \subst{\vec{v}} ) \implies t \holds \secprop ] \\
                      {}\land{} &\del[\Big][ \oskip\And[r]{ i \in [n] } \Forall{ \inst } \del[\big]( t \holds \ct_i \inst \implies t \holds \neg \secprop ) ] 
                    \end{aligned} \\
                    &\ftequiv \begin{aligned}[t]
                                &\del[\Big]( t \holds \del[\Big][ \del[\Big]( \oskip\And[r]{ i \in [n] } \NExists{ \vec{v} } \ct_i \subst{ \vec{v} } ) \impliestp \secprop ] ) \\
                      {}\land{} &\And*{ i \in [n] } \del[\Big]( t \holds \del[\Big][ \Forall{ \vec{v} } \ct_i \subst{ \vec{v} } \impliestp \neg \secprop ] )
                    \end{aligned} \\
                    &\ftequiv t \holds \TraceProp[\big]{ \VerE_{\secprop}^{\tp} } \land \And*{ i \in [n] } t \holds \TraceProp[\big]{ \VerNE_{\secprop,\ct_i}^{\tp} } \\
                    &\isequiv \ps{P} \holds^{\forall} \TraceProp[\big]{ \VerE_{\secprop}^{\tp} } \land \And*{ i \in [n] } \ps{P} \holds^{\forall} \TraceProp[\big]{ \VerNE_{\secprop,\ct_i}^{\tp} } \\
                    &\isequiv \VerE_{\secprop}^{\tp} \land \VerNE_{\secprop}^{\tp} \tag*{\qedhere}
  \end{align*}
  \endgroup
\end{proof}
\end{full}
\begin{lemma}[Minimality]
  \label{lem:min-equiv}
  \begin{equation*}
    \hyperref[cnd:min]{\Min} \iff \hyperref[cnd:mintp]{\Min^{\tp}}
  \end{equation*}
\end{lemma}
\begin{full}
\begin{proof}
  \begingroup
  \allowdisplaybreaks[1]
  \begin{align*}
    \Min &\isequiv \Forall{ t, S } S \in \vf(t) \implies \begin{aligned}[t]
             \NExists{ \qs{S'} } &S' \in \vf(t) \\
                       {}\land{} &S' \subsetneq S
         \end{aligned} \\
         &\isequiv \Forall{ t, i, \inst} t \holds \ct_i \inst \implies \begin{aligned}[t]
             \NExists{ j, \qs{\inst'} } &t \holds \ct_j \inst' \\
                              {}\land{} &\fv(\ct_j) \inst' \subsetneq \fv(\ct_i) \inst
         \end{aligned} \\
         &\isequiv \And*{ i \in [n] } \Forall{ t, \inst} t \holds \ct_i \inst \implies \And*{ j \in [n] } \Bigg( \begin{aligned}
             \NExists{ \qs{\inst'} } &t \holds \ct_j \inst' \\
                           {}\land{} &\fv(\ct_j) \inst' \subsetneq \fv(\ct_i) \inst
         \end{aligned} \Bigg) \\
         &\isequiv \And*{ i \in [n] } \Forall{ t } t \holds \Bigg[ \Forall{ \vec{v} } \ct_i \subst{\vec{v}}  \impliestp \And*{ j \in [n] } \Bigg( \begin{aligned}
             \NExists{ \vec{w} } &\ct_j \subst{\vec{w}} \\
                       {}\land{} &\TraceProp[\big]{ \vec{w} \subsetneq \vec{v} }
         \end{aligned} \Bigg) \Bigg] \\
         &\isequiv \And*{ i \in [n] } \Forall{ t } t \holds \TraceProp[\big]{ \Min_{\ct_i}^{\tp} } \\
         &\isequiv \And*{ i \in [n] } \ps{P} \holds^{\forall} \TraceProp[\big]{ \Min_{\ct_i}^{\tp} } \\
         &\isequiv \Min^{\tp} \tag*{\qedhere}
  \end{align*}
  \endgroup
\end{proof}
\end{full}
\begin{lemma}[Uniqueness]
  \label{lem:uniq-equiv}
  \begin{equation*}
    \hyperref[cnd:uniq]{\Uniq} \iff \hyperref[cnd:uniqtp]{\Uniq^{\tp}}
  \end{equation*}
\end{lemma}
\begin{full}
\begin{proof}
  \begingroup
  \allowdisplaybreaks[1]
  \begin{align*}
    \Uniq &\isequiv \Forall{ t, S } S \in \vf(t) \implies S \subseteq \cor(t) \\
          &\isequiv \Forall{ t, i, \inst} t \holds \ct_i \inst\implies \fv(\ct_i) \inst\subseteq \cor(t) \\
          &\isequiv \And*{ i \in [n] } \Forall{ t } t \holds \del[\Big][ \Forall{ \vec{v} } \ct_i \subst{\vec{v}} \impliestp \TraceProp[\big]{ \vec{v} \subseteq \fact{Corrupted} } ] \\
          &\isequiv \And*{ i \in [n] } \Forall{ t } t \holds \TraceProp[\big]{ \Uniq_{\ct_i}^{\tp} } \\
          &\isequiv \And*{ i \in [n] } \ps{P} \holds^{\forall} \TraceProp[\big]{ \Uniq_{\ct_i}^{\tp} } \\
          &\isequiv \Uniq^{\tp} \tag*{\qedhere}
  \end{align*}
  \endgroup
\end{proof}
\end{full}
\begin{lemma}[Sufficiency---Soundness]
  \label{lem:suff-snd}
  \begin{equation*}
    \hyperref[cnd:sufftp]{\Suff^{\tp}} \land  \hyperref[cnd:uniqtp]{\Uniq^{\tp}} \land \hyperref[cnd:ins-inj]{\InsI} \land \hyperref[cnd:rep-prop]{\RepP} \implies \hyperref[cnd:suff]{\Suff}
  \end{equation*}
\end{lemma}
\begin{full}
\begin{proof}
  Let $S \in \vf(t)$.
  By \cref{cor:verdict-function} there exists a case test $\ct_i$ and an instantiation $\inst$ such that $t \holds \ct_i \inst$ with $\fv(\ct_i) \inst = S$.
  By \nameref{cnd:sufftp} there exists a trace $t'$ such that $\ctr(t') = \Set[\big]{ (\ct_i, \inst') }$ and $\cor(t') \subseteq \fv(\ct_i) \inst'$.
  Using $(\ct_i, \inst) \in \ctr(t)$ as a witness, from \nameref{cnd:ins-inj} and \nameref{cnd:rep-prop} follows the existence of a single-matched trace $t''$ with $\ctr(t'') = \Set[\big]{ (\ct_i, \inst) }$ and $\cor(t'') = \cor(t')(\inst \circ {\inst'}^{-1})$.
  From the latter follows $\cor(t'') \subseteq S$.
  It remains to show that $t$ and $t''$ are related.
  From \nameref{cnd:uniqtp} follows $S \subseteq \cor(t'')$ and thus $\cor(t'') = S$.
  From \nameref{cnd:uniqtp} follows $S \subseteq \cor(t)$.
  As all requirements are fulfilled, we can apply \nameref{cnd:rel-in} to obtain $\rel(t,t'')$.
  Hence,
  \begin{equation*}
    \vf(t'') = \Set{ S } \land \cor(t'') \subseteq S \land \rel(t,t'')\,,
  \end{equation*}
  and \hyperref[cnd:suff]{$\Suff_S$} holds.
  Since $S$ has been arbitrary, the same argument applies to all $S$, which shows \hyperref[cnd:suff]{$\Suff$}.
\end{proof}
\end{full}
\begin{lemma}[Sufficiency---Completeness]
  \label{lem:suff-cmpl}
  \begin{equation*}
    \hyperref[cnd:suff]{\Suff} \land \hyperref[cnd:single]{\SinM} \land \hyperref[cnd:ver]{\Ver} \implies \hyperref[cnd:sufftp]{\Suff^{\tp}}
  \end{equation*}
\end{lemma}
\begin{full}
\begin{proof}
  By \hyperref[cnd:single]{$\SinM$} there exists for each case test $\ct_i$ a single-matched trace $t$ such that $\ctr(t) = \Set[\big]{ (\ct_i, \inst) }$.
  Let $S = \fv(\ct_i) \inst$.
  From \cref{cor:verdict-function} follows $S \in \vf(t)$.
  By \hyperref[cnd:suff]{$\Suff_S$} there exists a trace $t'$ such that
  \begin{equation*}
    \vf(t') = \verdict{ S } \land \cor(t') \subseteq S \land \rel(t,t')\,.
  \end{equation*}
  From \hyperref[cnd:ver]{$\Ver$} follows $t \holds \neg \secprop$ and $t' \holds \neg \secprop$.
  Along with $\rel(t,t')$ follows by \nameref{cnd:rel-el} that $\ctr(t') \subseteq \ctr(t)$.
  From \cref{cor:verdict-function} and $S \in \vf(t')$ follows $\ctr(t') = \Set[\big]{ (\ct_i, \inst) }$.
  Hence, with \cref{eq:cor-sub-tp} follows
  \begin{align*}
       t' \holds \Exists{ \vec{v} } \ct_i\subst{\vec{v}} &\land \del[\Big][ \Forall{ \vec{w} } \ct_i \subst{\vec{w}} \impliestp \vec{w} = \vec{v} ] \\
                                                                        &\land \del[\Big][ \And[r]{ j \in [n] \setminus \Set{i} } \NExists{ \vec{x} } \ct_j \subst{\vec{x}} ] \\
                                                                        &\land \TraceProp[\big]{ \fact{Corrupted} \subseteq \vec{v} }\,.
  \end{align*}
  The same argument applies to each single-matched trace for which there exists at least one for each case test by \hyperref[cnd:single]{$\SinM$}.
  This shows \hyperref[cnd:sufftp]{$\Suff^{\tp}$}.
\end{proof}
\end{full}
\begin{lemma}[Completeness]
  \label{lem:comp}
  \begin{equation*}
    \hyperref[cnd:vertp-ne]{\VerNE_{\secprop}^{\tp}} \implies \hyperref[cnd:comp]{\Comp}
  \end{equation*}
\end{lemma}
\begin{full}
\begin{proof}
  From \cref{eq:conj-comp-1} follows the existence of a trace $t'$ such that
  \begin{equation*}
    \vf(t') = \verdict{S} \land \cor(t') \subseteq S \land \rel(t,t')
  \end{equation*}
  and from \cref{eq:conj-comp-4} follows $t \holds \neg \secprop$.
  From \cref{cor:verdict-function} follows the existence of a case test $\ct_i$ and instantiation $\inst$ such that $t' \holds \ct_i \inst$.
  And thus $t' \holds \neg \secprop$ by \nameref{cnd:vertp-ne}.
  As all requirements are satisfied, \nameref{cnd:rel-el} can be applied to obtain $\ctr(t') \subset \ctr(t)$.
  From the latter follows with \cref{eq:ctr-vf} that $\vf(t') \subseteq \vf(t)$ and thus $S \in \vf(t)$.
\end{proof}
\end{full}

With the results above, we can proof the central theorems of this work---soundness and completeness of $\VC^{\tp}$.
\begin{theorem}[Soundness]
  \label{thm:snd-trace-prop}
  For any protocol $\ps{P}$, security property $\secprop$, and case tests $\CT = \ct_1, \dots, \ct_n$, if $\VC^{\tp}$, \hyperref[cnd:ins-inj]{$\InsI$}, and \hyperref[cnd:rep-prop]{$\RepP$} hold, then $\vfup_{\CT}$ provides $\ps{P}$ with accountability for~$\secprop$.
\end{theorem}
\begin{full}
\begin{proof}
  Assume $\VC^{\tp}$, \hyperref[cnd:ins-inj]{$\InsI$} and \hyperref[cnd:rep-prop]{$\RepP$} hold.
  From \cref{lem:ver-equiv,lem:min-equiv,lem:uniq-equiv,lem:suff-snd} follows $\VC$.
  By \cref{thm:snd-set} $\vf$ provides $\ps{P}$ with accountability for $\secprop$.
\end{proof}
\end{full}
\begin{theorem}[Completeness]
  \label{thm:cmpl-trace-prop}
  For any protocol $\ps{P}$, security property $\secprop$, and case tests $\CT = \ct_1, \dots, \ct_n$, if $\vfup_{\CT}$ provides $\ps{P}$ with accountability for $\secprop$, and \hyperref[cnd:single]{$\SinM$} holds, then $\VC^{\tp}$ holds.
\end{theorem}
\begin{full}
\begin{proof}
  Assume $\vf$ provides $\ps{P}$ with accountability for $\secprop$, and \hyperref[cnd:single]{$\SinM$} holds.
  From \cref{thm:cmpl-set} follows $\VC$.
  By \cref{lem:ver-equiv,lem:min-equiv,lem:uniq-equiv,lem:suff-cmpl} follows $\VC^{\tp}$.
\end{proof}
\end{full}

\subsection{Proof of sufficiency condition for \texorpdfstring{\hyperref[cnd:br]{$\BR$}}{BR}}
\label{sec:br-sufficiency-condition}

\begin{full}
Arguments about syntactic conditions inherently depend on the
calculus. However, the high-level argument is the same for both SAPiC
and multiset-rewrite rules. Hence we first provide a proof sketch, and then fully
elaborate the proof for multiset-rewrite rules.
\end{full}
\begin{conf}
  For the proof see \appendixorfull{sec:br-sufficiency-condition}.
\end{conf}

\begin{lemma}
    \label{lem:br-sufficiency}
    \begin{equation}
        \label{eq:br-sufficiency}
        \Parties \subseteq \PN \land \fn(\ps{P}) \cap \PN = \emptyset \implies \hyperref[cnd:br]{\BR}
    \end{equation}
\end{lemma}
\begin{full}
\begin{proofsketch}
    Assume \cref{eq:br-sufficiency} does not hold.
    Then there exists a bijection $\sub \colon \Parties \leftrightarrow \Parties$ and a trace $t \in \traces(\ps{P})$ such that $t \sub \notin \traces(\ps{P})$.
    Since $t \neq_E t \sub$, there exists w.l.o.g.\ an action $\fact{F}(m_1, \dots, m_n)@k \in t$ for messages $m_i$ such that $\fact{F}(m_1, \dots, m_n)@k \neq_E \fact{F}(m_1 \sub, \dots, m_n \sub)@k$.
    Hence, there exists a $j \in [n]$ such that $m_j \neq_E m_j \sub$.
    From $\Parties \subseteq \PN$ and the fact that $\sub$ is a bijection on $\PN$, it follows that $m_j$ and thus $m_j \sub$ contain public names.

    Since $\ps{P}$ does not contain public names by assumption, the public names cannot be hardcoded and must arise from variable realizations.
    For the same reason,
    these variables can be compared to other messages, but not to public names, hence
    the comparison results must be preserved under a bijective renaming of public names.

    Thus, whenever a message $m_i$ can be constructed in $\ps{P}$, the message $m_i \sub$ under the bijection $\sub$ can be constructed.
    This argument extends to all messages in the action $\fact{F}$ and to all actions in $t$.
    Hence, when $t \in \traces(\ps{P})$ then $t \sub \in \traces(\ps{P})$ which violates our assumption that \hyperref[cnd:br]{$\BR$} does not hold.
\end{proofsketch}

In the following, we assume that $\ps{P}$ is defined by a set of multiset-rewrite rules $\Set{ \ru_1, \dots, \ru_n }$.
\begin{proof}
    Assume
    \begin{enumerate*}[label=(\arabic*),ref=(\arabic*)]
        \item \label{it:parties-pubnames} $\Parties \subseteq \PN$ and
        \item \label{it:no-pubnames} $\fn(\ps{P}) \cap \PN = \emptyset$.
    \end{enumerate*}
    Let $\sub \colon \Parties \leftrightarrow \Parties$ be an arbitrary bijection and $t \in \traces(\ps{P})$ be an arbitrary trace.
    Each action $\fact{F}(m_1 \omega, \dots, m_k \omega)@k \in t$ corresponds to the application of a realized rule $\ri_i = \ru_i \omega$, where $\fact{F}(m_1 \omega, \dots, m_k \omega)@k \in \ri_i.a$ and $\ri_i.a$ denotes the multiset of actions in rule $\ri_i$.
    Note that the domain of $\omega$ are variables and the domain of $\sub$ are public names.

    To prove that $t \sub \in \traces(\ps{P})$, it suffices to show that each action $\fact{F}((m_1 \omega) \sub, \dots, (m_k \omega) \sub)@k \in t \sub$ corresponds to the application of a realized rule $\ri_i \sub$.
    To this end, we show that when a rule $\ru_i = (l, a, r)$ is applicable in state $S$ with realization $\omega$ producing actions $a \omega$ and leading to state $S' = S \setminus l \omega \cup r \omega$, then $\ru_i$ is also applicable in state $S \sub$ with realization $\sub \circ \omega$ producing actions $(a \omega) \sub$ and leading to state $S' \sub$.\looseness=-1

    Note that due to assumptions \ref{it:parties-pubnames} and \ref{it:no-pubnames}, all public names in a realized rule $\ri_i$ correspond to variables in $\ru_i$, i.e., are subterms of $\omega(v)$ for some variable in $\ri_i$.
    Thus $\sub \circ \omega$ gives rise to a rule instance $\ri_i' = \ru_i (\sub \circ \omega)$.

    Since $\ri_i$ is applicable, we know that for each fact $f \in l \omega$, there exists a fact $f' \in_E S$.
    Due to assumptions \ref{it:parties-pubnames} and \ref{it:no-pubnames}, $f' \sub \in S \sub$ and thus $(l \omega) \sub \subseteq S \sub$.
    Hence, the rule $\ri_i'$ is applicable in state $S \sub$ leading to state 
    \[
        S \sub \setminus ((l \omega) \sub) \cup ((r \omega ) \sub) = (S \setminus l \omega \cup r \omega) \sub = S' \sub\,.
    \]
    Thus, $t \sub \in \traces(\ps{P})$ and $t \sub$ is generated by the same sequence of multiset-rewrite rules as $t$ where each rule application is substituted by $\sub$. 
\end{proof}

\end{full}

\subsection{Implications of the Results}
\label{sec:inter-results}

\begin{figure}
  \caption[Decision diagram]{Decision diagram for the requirements and verification conditions defined in \cref{sec:verification-conditions-trace-prop}. Each edge represents an implication, each branch a disjunction.}
  \label{fig:decision-diagram}
  \centering
  \resizebox{\linewidth}{!}{
  \begin{tikzpicture}
    \GraphInit[vstyle=Empty]
    \SetVertexNormal[Shape = rectangle, LineWidth  = 0.4pt,LineColor  = white]
    \tikzset{EdgeStyle/.style = {-stealth, line width=0.8pt}}
    
    \Vertex[x=3.75, y=0, L=$\neg \VC^{\tp}$]{VCtp}
    \Vertex[x=6, y=0, L={\hyperref[cnd:single]{$\neg \SinM$}}]{SM}
    \Vertex[x=9, y=0, L={\hyperref[cnd:ins-inj]{$\neg \InsI$}}]{II}
   
    \Vertex[x=0, y=-2, L={\hyperref[cnd:uniqtp]{$\neg \Uniq^{\tp}$}}]{Utp}
    \Vertex[x=1.5, y=-2, L={\hyperref[cnd:mintp]{$\neg \Min^{\tp}$}}]{Mtp}
    \Vertex[x=3, y=-2, L={\hyperref[cnd:vertp-e]{$\neg \VerE_{\secprop}^{\tp}$}}]{VEtp}
    \Vertex[x=4.5, y=-2, L={\hyperref[cnd:vertp-ne]{$\neg \VerNE_{\secprop}^{\tp}$}}]{VNEtp}
    \Vertex[x=6, y=-2, L={\hyperref[cnd:sufftp]{$\neg \Suff^{\tp}$}}]{SFtp}
    \Vertex[x=9, y=-2, L={\hyperref[cnd:rep-prop]{$\neg \RepP$}}]{RP}
    
    \Vertex[x=0, y=-4, L={\hyperref[cnd:uniq]{$\neg \Uniq$}}]{U}
    \Vertex[x=1.5, y=-4, L={\hyperref[cnd:min]{$\neg \Min$}}]{M}
    \Vertex[x=3.75, y=-4, L={\hyperref[cnd:ver]{$\neg \Ver_{\secprop}$}}]{V}
    \Vertex[x=6, y=-4, L={\hyperref[cnd:suff]{$\neg \Suff$}}]{SF}
  
    \Vertex[x=7.5, y=-4, L={\hyperref[cnd:single]{$\neg \SinM$}}]{SMsf}
    \Vertex[x=3.75, y=-6, L=$\neg \VC$]{VC}
    
    \Vertex[x=3.75, y=-8, L={\hyperref[def:acc]{$\neg \Acc$}}]{ACC}
    \Vertex[x=8.25, y=-8, L={\hyperref[def:acc]{$\neg \Acc \mathbin{/} \Acc$}}]{unACC}
    
    \tikzset{VertexStyle/.append style ={inner xsep = 8pt}}
    \Vertex[x=3.75, y=-6, L=$\neg \VC$]{VC}
  
    \Edges(SM,SFtp)
    \Edges(VCtp,Utp,U,VC,ACC)
    \Edges(VCtp,Mtp,M,VC,ACC)
    \Edges(VCtp,VEtp,V,VC,ACC)
    \Edges(VCtp,VNEtp,V,VC,ACC)
    \Edges(VCtp,SFtp,SF,VC,ACC)
    \Edges(SFtp,SMsf,unACC)
    \Edges(II,RP,unACC)
  \end{tikzpicture}
  }
\end{figure}
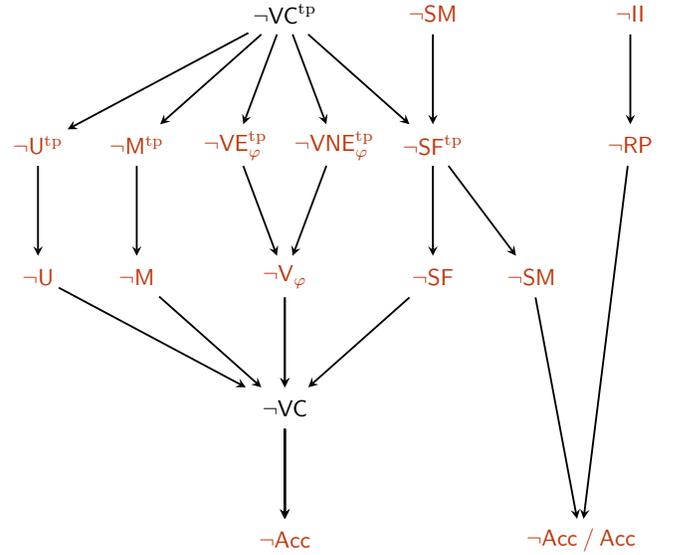

The implications of a failed condition are depicted in \cref{fig:decision-diagram}.
Each arrow in the diagram represents an implication; each branch a disjunction.
The implications follow from \cref{lem:suff-cmpl,lem:ver-equiv,lem:min-equiv,lem:uniq-equiv}, and \cref{def:acc} as well as the definitions of the respective conditions.
For example, if \hyperref[cnd:vertp-ne]{$\neg \VerNE_{\secprop}^{\tp}$}, we know by \cref{lem:ver-equiv} that \hyperref[cnd:ver]{$\neg \Ver_{\secprop}$}.
Hence, $\neg \VC$ and by \cref{thm:cmpl-set} accountability is not provided.

In the following, we discuss the meaning of each failed condition and give hints on how to fix the problems.
\paragraph{Case {\hyperref[cnd:sufftp]{$\neg \Suff_{\ct_i}^{\tp}$}}}
  There does not exist a single-matched trace for $\ct_i$ in which only a subset of the blamed parties is corrupted.
  At least one party, which is needed to cause a violation is not blamed.
  Accountability may still be provided.

  \textbf{Hint:} Assume \nameref{cnd:vertp-ne}.
  If $\neg \nameref{cnd:single}$, we should solve this problem first.
  In all single-matched traces of $\ct_i$, there exists at least one corrupted party which is not one of the instantiated free variables of $\ct_i$.
  It may be possible to revise $\ct_i$ by adding additional free variables and action constraints such that all parties needed for a violation are blamed by $\ct_i$.

\paragraph{Case {\hyperref[cnd:vertp-e]{$\neg \VerE_{\secprop,\ct_i}^{\tp}$}}}
  No case test holds, but the security property is violated.
  This indicates that the case tests are not exhaustive, that is, capture all possible ways to cause a violation.
  Accountability is not provided.

  \textbf{Hint:} The trace \Tamarin{} found as a counterexample may give a hint for an additional case test or shows that the security can be violated in an unintended way.

\paragraph{Case {\hyperref[cnd:vertp-ne]{$\neg \VerNE_{\secprop,\ct_i}^{\tp}$}}}
  The case test $\ct_i$ holds but the security property is not violated.
  This indicates that there exists a trace where $\ct_i$ is not sufficient to cause a violation.
  Accountability is not provided.

  \textbf{Hint:} The trace \Tamarin{} found as a counterexample may give a hint to revise $\ct_i$ such that for all traces in which it holds the security property is violated.

\paragraph{Case {\hyperref[cnd:mintp]{$\neg \Min_{\ct_i}^{\tp}$}}}
  There exists an instantiation of a case test $\ct_j$ which blames strictly fewer parties than an instantiation of $\ct_i$ in the same trace.
  Accountability is not provided.

  \textbf{Hint:} Assume \hyperref[cnd:vertp-ne]{$\VerNE_{\ct_i}^{\tp}$} and \hyperref[cnd:vertp-ne]{$\VerNE_{\ct_j}^{\tp}$}.
   If both $\ct_i$ and $\ct_j$ are necessary for \nameref{cnd:vertp-e} to hold, they need to be separated such that they do not hold simultaneously.
   This can be accomplished by replacing $\ct_i$ with $\ct_i \land \neg \del[\big]( \ct_j \land \TraceProp[\big]{ \fv(\ct_j) \subsetneq \fv(\ct_i) } )$.

\paragraph{Case {\hyperref[cnd:uniqtp]{$\neg \Uniq_{\ct_i}^{\tp}$}}}
  A party is blamed by an instantiation of $\ct_i$ but it has not been corrupted, thereby holding an honest party accountable.
  Accountability is not provided.

  \textbf{Hint:} Assume \nameref{cnd:vertp-ne}.
  If $\neg \nameref{cnd:mintp}$, we should solve this problem first.
  The trace \Tamarin{} found as a counterexample shows which party is blamed without having been corrupted.
  If the corresponding instantiated free variable can never be corrupted, it can be quantified in $\ct_i$ to avoid being blamed.
  If it can be corrupted for some traces, a closer look on $\ct_i$ and the protocol is necessary.

\paragraph{Case {\hyperref[cnd:single]{$\neg \SinM_{\ct_i}^{\tp}$}}}
  There does not exist a single-matched trace for $\ct_i$.
  Either
  \begin{enumerate}[label=(\roman*)]
    \item\label[case]{it:dec-relt-1} there does not exist a trace where $\ct_i$ holds, or
    \item\label[case]{it:dec-relt-2} $\ct_i$ always holds with multiple instantiations, or
    \item\label[case]{it:dec-relt-3} for all traces there exist another case test which holds at the same time
  \end{enumerate}
  Accountability may still be provided.

  \textbf{Hint:} Assume \nameref{cnd:vertp-ne}.
  In \cref{it:dec-relt-1}, $\ct_i$ may be ill-defined or contains a logic error.
  In \cref{it:dec-relt-2}, if all the instantiations are permutations of each other, a single-matched trace may be obtained by making $\ct_i$ antisymmetric.
  Then for all instantiations $\inst$, $\inst$'
  \begin{equation*}
    t \holds \ct_i \inst \land t \holds \ct_i \inst' \land \fv(\ct_i) \inst = \fv(ct_i) \inst' \implies \inst = \inst'\,.
  \end{equation*}
  If the instantiations are not permutations, at least two disjoint groups of parties are always blamed.
  This requires a closer look on $\ct_i$ and the protocol.
  In \cref{it:dec-relt-3}, it may be possible to merge multiple case tests together for which then a single-matched trace exists.

\paragraph{Case {\hyperref[cnd:ins-inj]{$\neg \InsI_{\ct_i}^{\tp}$}}}
  The case test $\ct_i$ is not injective.
  There exists an instantiation mapping distinct free variables to the same party.
  Accountability may still be provided.

  \textbf{Hint:} See \cref{ex:sub-inj-split} for a way to split $\ct_i$.

We note that for the conditions \nameref{cnd:sufftp}, \nameref{cnd:mintp}, \nameref{cnd:uniqtp}, and \nameref{cnd:single}, we assumed above that the case tests satisfy \nameref{cnd:vertp-ne}.
If this is not the case, then the case test has a fatal error---it does not always lead to a violation---which renders the other conditions meaningless.

\begin{full}
\subsection{Stateful Applied Pi Calculus}
\label{sec:sapic}

In this \lcnamecref{sec:sapic}, we introduce the Stateful Applied Pi Calculus---called SAPiC \parencite{kremer2014automated,backes2017novel}---which is an extension to the well-known applied-$\pi$ calculus \parencite{abadi2001mobile}.
In addition to the functionality of the former calculus, SAPiC provides support for accessing and updating an explicit global state.

\begin{figure}
  \caption{SAPiC syntax}
  \label{fig:sapic-syntax}
  \centering
  \begin{minipage}{.49\linewidth}
    \begin{grammar}
      <$\ps{P}$,$\ps{Q}$> ::= $0$
      \alt $\ps{P} \pcp \ps{Q}$
      \alt $\ps{P} + \ps{Q}$
      \alt $\rpl \ps{P}$
      \alt $\new \var{n};\ \ps{P}$
      \alt $\kw{out}([\var{M}{,}]\, \var{N});\ \ps{P}$
      \alt $\kw{in}([\var{M}{,}]\, \var{N});\ \ps{P}$
      \alt \begin{spreadlines}{0pt}$\begin{aligned}[t]&\kw{if}\ \var{M} = \var{N}\ \kw{then}\ \ps{P}\\ &[\kw{else}\ \ps{Q}]\end{aligned}$\end{spreadlines}\vspace{1.5pt}
      \alt $\kw{event}\ \fact{F};\ \ps{P}$
      \alt $\kw{insert}\ \var{M}, \var{N};\ \ps{P}$
      \alt $\kw{delete}\ \var{M};\ \ps{P}$
      \alt \begin{spreadlines}{0pt}$\begin{aligned}[t]&\kw{lookup}\ \var{M}\ \kw{as}\ \var{x}\ \kw{in}\ \ps{P}\\ &[\kw{else}\ \ps{Q}]\end{aligned}$\end{spreadlines}\vspace{1.5pt}
      \alt $\kw{lock}\ \var{M};\ \ps{P}$
      \alt $\kw{unlock}\ \var{M};\ \ps{P}$
    \end{grammar}
  \end{minipage}%
  \begin{minipage}{.51\linewidth}
    \begin{grammar}
      <$\var{M}, \var{N}$> ::= $\var{x}$, $\var{y}$, $\var{z} \in \Vars$
      \alt $\var{p} \in PN$
      \alt $\var{n} \in FN$
      \alt $\fun{f}(\var{M_1}, \dots, \var{M_k})$, $f \in \Sig^k$
    \end{grammar}
  \end{minipage}
\end{figure}

In the following, we explain the syntax and semantics of the calculus.
The syntax of SAPiC is shown in \cref{fig:sapic-syntax}.

\medskip

\noindent $0$:

The terminal process.

\medskip

\noindent $\ps{P} \pcp \ps{Q}$:

The parallel execution of the processes $\ps{P}$ and $\ps{Q}$.

\medskip

\noindent $\ps{P} + \ps{Q}$

\emph{External} non-deterministic choice.
If $\ps{P}$ or $\ps{Q}$ can reduce to a process $\ps{P}'$ or $\ps{Q}'$, $\ps{P} + \ps{Q}$ may reduce to either.

\medskip
\noindent $\rpl \ps{P}$

The replication of $\ps{P}$ allowing an unbounded number of sessions in protocol executions.
It is equivalent to $\ps{P} \pcp \rpl \ps{P}$.

\medskip

\noindent $\new \var{n};\ \ps{P}$
\nopagebreak

This construct binds the name $\var{n} \in \FN$ in $\ps{P}$ and models the generation of a fresh, random value.

\medskip

\noindent \begin{spreadlines}{1pt}
     $\begin{aligned}[t]
       &\kw{out}\del( [\var{M}{,}]\, \var{N} );\ \ps{P} \\
       &\kw{in}\del( [\var{M}{,}]\, \var{N} );\ \ps{P}
     \end{aligned}$
   \end{spreadlines}

These constructs represent the output and input of a message $\var{N}$ on channel $\var{M}$ respectively.
The channel argument is optional and defaults to the public channel $c$.
In contrast to the applied-$\pi$ calculus \parencite{abadi2001mobile}, SAPiC's input construct performs pattern matching instead of variable binding.

\medskip
  
\noindent $\kw{if}\ \var{M} = \var{N}\ \kw{then}\ \ps{P}\ [\kw{else}\ \ps{Q}]$

If $\ps{M} =_E \ps{N}$, this process reduces to $\ps{P}$ and otherwise to $\ps{Q}$.
The else-branch is optional and defaults to the $0$ process.

\medskip
 
\noindent $\kw{event}\ \fact{F};\ \ps{P}$

This construct leaves the fact $\fact{F}$ in the trace of the process execution which is useful in the definition of trace formulas.

\medskip

\noindent $\kw{insert}\ \var{M}, \var{N};\ \ps{P}$

This construct associates the key $\var{M}$ with the value $\var{N}$. 
An insert to an existing key overwrites the old value.

\medskip

\noindent $\kw{delete}\ \var{M};\ \ps{P}$

This construct removes the value associated to the key $\var{M}$.

\medskip

\noindent $\kw{lookup}\ \var{M}\ \kw{as}\ \var{x}\ \kw{in}\ \ps{P}\ [\kw{else}\ \ps{Q}]$

This construct retrieves the value associated to the key $\var{M}$ and binds it to $x$ in $\ps{P}$.
If no value has been associated with $\var{M}$, it reduces to $\ps{Q}$.

\medskip

\noindent \begin{spreadlines}{0pt}
   $\begin{aligned}[t]
      &\kw{lock}\ \var{M};\ \ps{P} \\
      &\kw{unlock}\ \var{M};\ \ps{P}
    \end{aligned}$
  \end{spreadlines}

These constructs protect a term $\var{M}$ from concurrent access similar to Dijkstra's binary semaphores.
If $\var{M}$ has been locked, any subsequent attempt to lock $\var{M}$ will be blocked until $\var{M}$ has been unlocked.
This is important if parallel processes read and modify shared state.

\medskip

\paragraph{Frames and deduction}
During a protocol execution, the adversary may compute new messages from observed ones.
This is formalized by a \emph{deduction relation} and a \emph{frame}.
A frame, denoted by $\new \widetilde{\var{n}}. \sub$, consists of a set of fresh names $\widetilde{\var{n}}$ and a substitution $\sub$.
The fresh names are the secrets generated by the protocol which are \emph{a priori} unknown to the adversary and the substitution represents the observed messages.
The deduction rules of \cref{fig:deduction-rel} allow the adversary
\begin{itemize}
  \item to learn a free or public name if it is not a secret (\textsc{Dname}),
  \item to obtain a message in the substitution (\textsc{Dframe}),
  \item to derive a term equal modulo $E$ to an already deduced term (\textsc{Deq}), or
  \item to apply a non-private function to already deduced terms (\textsc{Dappl}).
\end{itemize}
\begin{definition}[Deduction]
  The deduction relation $\new \widetilde{\var{n}}. \sub \vdash t$ is defined as the smallest relation between frames and terms according to the deduction rules in \cref{fig:deduction-rel}.
\end{definition}

\begin{figure}
  \caption{Deduction relation}
  \label{fig:deduction-rel}
  \centering
  \begin{mathpar}
    \inferrule*[right=Dname, vcenter]
      {\var{a} \in \FN \\ \var{a} \notin \widetilde{\var{n}}}
      {\new \widetilde{\var{n}}. \sub \vdash \var{a}}

    \inferrule*[right=Deq, vcenter]
      {\new \widetilde{\var{n}}. \sub \vdash \var{t} \\ \var{t} =_E \var{t'}}
      {\new \widetilde{\var{n}}. \sub \vdash \var{t'}}

    \inferrule*[right=Dframe, vcenter]
      {\var{x} \in \dom(\sub)}
      {\new \widetilde{\var{n}}. \sub \vdash \var{x} \sub}

    \inferrule*[right=Dappl, vcenter]
      {\new \widetilde{\var{n}}. \sub \vdash \var{t_1}\ \dots\ \new \widetilde{\var{n}}. \sub \vdash \var{t_k} \\ \fun{f} \in \Sig^k \setminus \Sig_{\text{priv}}^k}
      {\new \widetilde{\var{n}}. \sub \vdash \fun{f}(\var{t_1}, \dots, \var{t_k})}
  \end{mathpar}
\end{figure}

\paragraph{Operational semantics}
The semantics of SAPiC is defined by a labeled transition relation between \emph{process configurations}.
A process configuration is a 5-tuple $(\FNames, \Store, \GProcs, \sub, \Locks)$ where
\begin{itemize}
  \item $\FNames \in \FN$ is the set of fresh names generated by the processes;
  \item $\Store \colon \GTerms_{\Sig} \to \GTerms_{\Sig}$ is a partial function modeling the store;
  \item $\GProcs$ is a multiset of ground processes representing the processes executed in parallel;
  \item $\sub$ is a ground substitution modeling the messages output to the environment;
  \item $\Locks \subseteq \GTerms_{\Sig}$ is the set of currently active locks.
\end{itemize}

The transition relation is specified by the rules shown in \cref{tab:sapic-trans-rel}.
Transitions are labeled by sets of ground facts.
Reducing a process means that it can transition from a configuration $c$ to a configuration $c'$ with a set of facts $\Set{\fact{F_1}, \dots, \fact{F_n} }$  which is denoted by $c \transin*{\Set{\fact{F_1}, \dots, \fact{F_n} }} c'$.
Empty sets and brackets around singleton sets are omitted for clarity.
We write $\transin{}$ for $\transin{\emptyset}$ and $\transin{\fact{f}}$ for $\transin{\Set{\fact{f}}}$.
An \emph{execution} is a sequence of consecutive configurations $c_1 \transin{\fact{F_1}} \dots \transin{\fact{F_n}} c_n$.
The \emph{trace} of an execution is the sequence of nonempty facts $\fact{F_i}$.
The reflexive transitive closure of $\transin{}$, which are the transitions labeled by the empty sets, is denoted by $\transin{}^*$ and $\transc{\fact{f}}$ denotes $\transin{}^* \transin{\fact{f}} \transin{}^*$.
The set of traces of a process $\ps{P}$ contains all traces of possible executions of the process.
\begin{definition}[Traces of $\ps{P}$]
  Given a ground process $\ps{P}$, the traces of $\ps{P}$ are defined by
  \begin{align*}
    \traces \del[\big](\ps{P}) = \DisplaySet[\Big]{ \del[\big]( \fact{F_1}, \dots, \fact{F_n} ) \given c_0 \transc{\fact{F_1}} \dots \transc{\fact{F_n}} c_n }\,,
  \end{align*}
  where $c_0 = \del[\big]( \emptyset, \emptyset, \Set{ \ps{P} }, \emptyset, \emptyset )$ is the initial process configuration.
\end{definition}

\begin{table*}
  \caption{Operational semantics of SAPiC}
  \label{tab:sapic-trans-rel}
  \footnotesize
  \centering
  \begin{tabular}{l@{}c@{\hspace{.5\tabcolsep}}l@{\hspace{.75\tabcolsep}}l}
    \toprule
    current configuration ($c_i$) & label & \multicolumn{2}{@{}l}{next configuration ($c_{i+1}$)} \\
    \midrule
    $\GProcs \munion \Set{ 0 }$                  & $\trans{}$ & $\GProcs$ & \\ \midrule
    $\GProcs \munion \Set{ \ps{P} \pcp \ps{Q} }$ & $\trans{}$ & $\GProcs \munion \Set{ \ps{P}, \ps{Q} }$ & \\ \midrule
    $\GProcs \munion \Set{ \rpl \ps{P} }$        & $\trans{}$ & $\GProcs \munion \Set{ \ps{P}, \rpl \ps{P} }$ & \\ \midrule
    $\GProcs \munion \Set{ \new n;\ \ps{P} }$     & $\trans{}$ & $\begin{aligned}
      &\FNames \union \Set{ \var{n'} } \\
      &\GProcs \munion \Set{ \ps{P}\Set{ \sfrac{\var{n'}}{\var{n}} } }
    \end{aligned}$ & if $\var{n'}$ is fresh \\ \midrule
    $\GProcs$                                    & $\trans*{\fact{K}(\var{M})}$ & $\GProcs$ & if $\new \FNames. \sub \vdash \var{M}$ \\ \midrule
    $\GProcs \munion \Set{ \kw{out}\del[\big]( \var{M}, \var{N} );\ \ps{P} }$ & $\trans*{\fact{K}(\var{M})}$ & $\begin{aligned}
      &\GProcs \munion \Set{ \ps{P} } \\
      &\sub \union \Set{ \sfrac{ \var{N} }{ \var{x} } }
    \end{aligned}$ & $\begin{aligned}
      &\text{if $\var{x}$ is fresh}\\
      &\text{and $\new \FNames. \sub \vdash \var{M}$}
    \end{aligned}$ \\ \midrule
    $\GProcs \munion \Set{ \kw{in}\del[\big]( \var{M}, \var{N} );\ \ps{P} }$ & $\trans*{\fact{K}(\del< \var{M}, \var{N} \gamma >)}$ & $\GProcs \munion \Set{ \ps{P} \gamma }$ & $\begin{aligned}
      &\text{if $\new \FNames. \sub \vdash \var{M}$, $\new \FNames. \sub \vdash \var{N} \gamma$} \\
      &\text{and $\gamma$ is grounding for $\var{N}$}
    \end{aligned}$ \\ \midrule
    $\GProcs \munion \Set{ \kw{out}\del[\big]( \var{M}, \var{N} );\ \ps{P}, \kw{in}\del[\big]( \var{M'}, \var{N'} );\ \ps{Q} }$ & $\trans{}$ & $\GProcs \munion \Set{ \ps{P}, \ps{Q} \gamma }$ & $\begin{aligned}
      &\text{if $\var{M} =_E \var{M'}$ and $\var{N} =_E \var{N'} \gamma$}\\
      &\text{and $\gamma$ is grounding for $\var{N'}$}
    \end{aligned}$ \\ \midrule
    \multirow{2}{*}[-0.5em]{$\GProcs \munion \Set{ \kw{if}\ \var{pr}( \var{M_1}, \dots, \var{M_n} )\ \kw{then}\ \ps{P}\ \kw{else}\ \ps{Q} }$} & \multirow{2}{*}[-0.5ex]{$\trans{}$} & $\GProcs \munion \Set{ \ps{P} }$ & if $\phi_{\var{pr}} \Set{ \sfrac{\var{M_1}}{\var{x_1}}, \dots, \sfrac{\var{M_n}}{\var{x_n}} }$ \\ \cmidrule{3-4}
    & & $\GProcs \munion \Set{ \ps{Q} }$ & otherwise \\ \midrule
    $\GProcs \munion \Set{ \kw{event}\ \fact{F};\ \ps{P} }$ & $\trans{\fact{F}}$ & $\GProcs$ & \\ \midrule
    $\GProcs \munion \Set{ \kw{insert}\ \var{M}, \var{N};\ \ps{P} }$ & $\trans{}$ & $\begin{aligned}
      &\Store \del[\big][ \var{M} \mapsto \var{N} ] \\
      &\GProcs \munion \Set{ \ps{P} }
    \end{aligned}$ & \\ \midrule
    $\GProcs \munion \Set{ \kw{delete}\ \var{M};\ \ps{P} }$ & $\trans{}$ & $\begin{aligned}
      &\Store \del[\big][ \var{M} \mapsto \bot ] \\
      &\GProcs \munion \Set{ \ps{P} }
    \end{aligned}$ & \\ \midrule
    \multirow{2}{*}[-1.1em]{$\GProcs \munion \Set{ \kw{lookup}\ \var{M}\ \kw{as}\ \var{x}\ \kw{in}\ \ps{P}\ \kw{else}\ \ps{Q};\ \ps{P} }$} & \multirow{2}{*}[-1.1em]{$\trans{}$} & $\GProcs \munion \Set{ \ps{P} \Set{ \sfrac{\var{V}}{\var{x}} } }$ & $\begin{aligned}
      &\text{if $\Store(\var{N}) =_E V$ is defined}\\
      &\text{and $\var{N} =_E \var{M}$}
    \end{aligned}$ \\ \cmidrule{3-4}
    & & $\GProcs \munion \Set{ \ps{Q} }$ & $\begin{aligned}
      &\text{if $\Store(\var{N})$ is undefined}\\
      &\text{for all $\var{N} =_E \var{M}$}
    \end{aligned}$ \\ \midrule
    $\GProcs \munion \Set{ \kw{lock}\ \var{M};\ \ps{P} }$ & $\trans{}$ & $\begin{aligned}
      &\GProcs \munion \Set{ \ps{P} } \\
      &\Locks \union \Set{ \var{M} }
    \end{aligned}$ & if $\var{M} \notin_E \Locks$ \\ \midrule
    $\GProcs \munion \Set{ \kw{unlock}\ \var{M};\ \ps{P} }$ & $\trans{}$ & $\begin{aligned}
      &\GProcs \munion \Set{ \ps{P} } \\
      &\mathrlap{\Locks \setminus \Set{ \var{M'} \given \var{M'} =_E \var{M} }}
    \end{aligned}$ & \\
    \bottomrule
  \end{tabular}
\end{table*}

\subsection{Ensuring Guardedness}
\label{sec:guardedness}

Trace properties in \Tamarin{} are specified by trace formulas from a guarded fragment of two-sorted first-order logic.
Guardedness imposes requirements on the structure of the trace formulas.
Universally and existentially quantified variables have to be \emph{guarded} by an action constraint directly after the quantifier in which all the variables occur.
For universally quantified trace formulas, the outermost logical operator inside the quantifier has to be an implication; for existentially quantified trace formulas a conjunction.
Formally, we can define guardedness as follows.
\begin{definition}[Guarded trace formula]
  A trace formula $\varphi$ is guarded if there exists a fact $\fact{Action} \in \Facts$ and a trace formula $\psi$ such that
  \begin{align}
    \varphi &= \Exists{ \vec{x}, i } \fact{Action}(\vec{x})@i \land \psi(\vec{x}) \text{ or} \label{eq:guarded-exists} \\
    \varphi &= \Forall{ \vec{x}, i } \fact{Action}(\vec{x})@i \impliestp \psi(\vec{x})\,. \label{eq:guarded-forall}
  \end{align}
\end{definition}

In order for \Tamarin{} to verify the generated lemmas, we have to ensure that they conform to either \cref{eq:guarded-exists} or \cref{eq:guarded-forall}.
Taking a closer look on the defined trace properties in \cref{sec:verification-conditions-trace-prop}, we see that a case test $\ct_i$ occurs in exactly three different kinds of subformulas.
\begin{align}
  &\Exists{ \vec{v} } \ct_i \subst{\vec{v}} \label{eq:cnd-ct-sub-1} \\
  &\Forall{ \vec{v} } \neg \ct_i \subst{\vec{v}} \label{eq:cnd-ct-sub-2} \\
  &\Forall{ \vec{v} } \ct_i \subst{\vec{v}} \impliestp \gamma(\vec{v}) \label{eq:cnd-ct-sub-3}
\end{align}

Assume $\ct_i$ is guarded and has the form of \cref{eq:guarded-exists}.
Expanding $\ct_i$ in the above formulas yields
\begingroup
\allowdisplaybreaks
\begin{gather*}
  \begin{aligned}
    &\Exists{ \vec{v} } \Exists{ \vec{x}, i } \fact{Action}(\vec{x})@i \land \del[\big](\psi(\vec{x}))\subst{\vec{v}} \\
    \isequiv &\Exists{ \vec{v}, \vec{x}, i } \fact{Action}(\vec{x})@i \land \del[\big](\psi(\vec{x}))\subst{\vec{v}}
  \end{aligned} \\[0.65em]
  \begin{aligned}
    &\Forall{ \vec{v} } \neg \del[\big]( \Exists{ \vec{x}, i } \fact{Action}(\vec{x})@i \land \del[\big](\psi(\vec{x}))\subst{\vec{v}} ) \\
    \isequiv &\Forall{ \vec{v}, \vec{x}, i } \fact{Action}(\vec{x})@i \impliestp \neg \del[\big](\psi(\vec{x}))\subst{\vec{v}}
  \end{aligned} \\[0.65em]
  \begin{aligned}
    &\Forall{ \vec{v} } \del[\big]( \Exists{ \vec{x}, i } \fact{Action}(\vec{x})@i \land \del[\big](\psi(\vec{x}))\subst{\vec{v}} ) \impliestp \gamma(\vec{v}) \\
    \isequiv &\Forall{ \vec{v}, \vec{x}, i } \fact{Action}(\vec{x})@i \land \del[\big](\psi(\vec{x}))\subst{\vec{v}} \impliestp \gamma(\vec{v})\,.
  \end{aligned}
\end{gather*}
\endgroup

All formulas are in the form of \cref{eq:guarded-exists} or \cref{eq:guarded-forall} and are thus themselves guarded.

Assume $\ct_i$ is guarded and has the form of \cref{eq:guarded-forall}.
Expanding $\ct_i$ in the above formulas yields
\begingroup
\allowdisplaybreaks
\begin{gather*}
  \begin{aligned}
    &\phantom{\isequiv} \Exists{ \vec{v} } \Forall{ \vec{x}, i } \fact{Action}(\vec{x})@i \impliestp \del[\big](\psi(\vec{x}))\subst{\vec{v}} 
  \end{aligned} \\[0.65em]
  \begin{aligned}
    &\Forall{ \vec{v} } \neg \del[\big]( \Forall{ \vec{x}, i } \fact{Action}(\vec{x})@i \impliestp \del[\big](\psi(\vec{x}))\subst{\vec{v}} ) \\
    \isequiv &\Forall{ \vec{v} } \Exists{ \vec{x}, i } \fact{Action}(\vec{x})@i \land \neg \del[\big](\psi(\vec{x}))\subst{\vec{v}}
  \end{aligned} \\[0.65em]
  \begin{aligned}
    &\Forall{ \vec{v} } \del[\big]( \Forall{ \vec{x}, i } \fact{Action}(\vec{x})@i \impliestp \del[\big](\psi(\vec{x}))\subst{\vec{v}} ) \impliestp \gamma(\vec{v}) \\
    \isequiv &\Forall{ \vec{v} } \Exists{ \vec{x}, i } \del[\big]( \neg \fact{Action}(\vec{x})@i \lor \del[\big](\psi(\vec{x}))\subst{\vec{v}} ) \impliestp \gamma(\vec{v})\,.
  \end{aligned}
\end{gather*}
\endgroup

We notice that none of the formulas has the form of \cref{eq:guarded-exists} or \cref{eq:guarded-forall}, since universal and existential quantifiers cannot be combined.
Trace formulas of this form are outside the guarded fragment that \Tamarin{} can verify.
However, a case test $\ct_i$ in the form of \cref{eq:guarded-forall} can be transformed into a case test $\ct'_i$ in the form of \cref{eq:guarded-exists} by adding guardedness constraints.
We have
\begin{equation*}
  \ct_i' \subst{\vec{v}} = \del[\Big][ \Exists{ k } \And*{ \ell \in \idx(\vec{v}) } \fact{Guarded}(\vec{v}_{\ell})@k ] \land \ct_i \subst{\vec{v}}\,,
\end{equation*}
which is guarded and in the form of \cref{eq:guarded-exists}.
If the protocol is adapted to issue $\fact{Guarded}$ facts for all parties in $\Parties$, then $\ct_i'$ is equivalent to $\ct_i$.

In summary, to ensure the guardedness of the generated lemmas, it is sufficient to require that all case tests are guarded and in the form of \cref{eq:guarded-exists}.
Since a guarded case test can always be transformed into this form, this is only a technical requirement and does not limit the expressiveness of the case tests.
Moreover, we note that this is a sufficient but not necessary condition.
There are case tests that are not guarded themselves but the generated lemmas are.
\end{full}

\end{document}